\documentclass{article}
\usepackage{fullpage}
\usepackage{times}
\usepackage{bm, bbm, amsmath,amsfonts,amssymb,amsthm}
\usepackage{enumerate}
\usepackage{comment}
\usepackage{xcolor}
\usepackage{xspace}
\usepackage{tikz}
\usepackage{paralist}
\usepackage{framed}
\usepackage{xparse}
\usepackage[colorlinks=true, allcolors=blue]{hyperref}
\usepackage{subcaption}
\usepackage{xparse}
\usepackage{algorithm}
\usepackage[noend]{algpseudocode}
\usetikzlibrary{decorations.pathreplacing, decorations.pathmorphing}

\newtheorem{theorem}{Theorem}
\newtheorem{lemma}[theorem]{Lemma}
\newtheorem{claim}[theorem]{Claim}
\newtheorem{fact}[theorem]{Fact}
\newtheorem{question}[theorem]{Question}
\newtheorem{proposition}[theorem]{Proposition}
\newtheorem{observation}[theorem]{Observation}
\newtheorem{corollary}[theorem]{Corollary}

\newtheorem{definition}[theorem]{Definition}
\newcommand{\remark}[1]{\noindent \textit{Remark.} {#1}}
\newenvironment{proofsketch}[1]{\noindent{\it Proof Sketch.} {#1}{\qed\bigskip}}
\allowdisplaybreaks

\newcommand{\E}{\mathop{\mathbb{E}}}

\newcommand{\cH}{\mathcal{H}}

\usepackage{algorithm}

\newcommand{\removelatexerror}{\let\@latex@error\@gobble}

\usepackage{cleveref}

\usepackage[backend=biber, style=alphabetic, doi=false, url=false, isbn=false, minalphanames = 5, maxalphanames = 10]{biblatex}
\title{Strong XOR Lemma for Information Complexity}
\author{ Pachara Sawettamalya \footnote{Department of Computer Science, Princeton University. \ \href{mailto:ps3122@princeton.edu}{\url{ps3122@princeton.edu}}}  \and Huacheng Yu\footnote{Department of Computer Science, Princeton University. \ \href{mailto:yuhch123@gmail.com}{\url{yuhch123@gmail.com}}}}
\date{}

\usepackage{thmtools}
\usepackage{thm-restate}

\usepackage{amsmath}
\usepackage{thm-restate}
\newcommand{\dkl}[2]{\textnormal{D}\left(\frac{#1}{#2}\right)}

\newcommand{\IC}[0]{\mathsf{IC}}

\newcommand{\M}[0]{\emph{\textnormal{M}}}
\newcommand{\adv}{\textnormal{adv}}
\newcommand{\stdz}{\mathsf{standardize}}
\newcommand{\D}[0]{\mathcal{D}}
\newcommand{\dklhor}[2]{\textnormal{D} \hspace{0.2mm} ( #1 \ \| \ #2 )}

\begin{document}

\maketitle
\begin{abstract}
\vspace{1mm}

For any $\{0,1\}$-valued function $f$, its \emph{$n$-folded XOR} is the function $f^{\oplus n}$ where $f^{\oplus n}(X_1, \ldots, X_n) = f(X_1) \oplus \cdots \oplus f(X_n)$. Given a procedure for computing the function $f$, one can apply a ``naive" approach to compute $f^{\oplus n}$ by computing each $f(X_i)$ independently, followed by XORing the outputs. This approach uses $n$ times the resources required for computing $f$.

In this paper, we prove a strong XOR lemma for \emph{information complexity} in the two-player randomized communication model: if computing $f$ with an error probability of $O(n^{-1})$ requires revealing $I$ bits of information about the players' inputs, then computing $f^{\oplus n}$ with a constant error requires revealing $\Omega(n) \cdot (I - 1 - o_n(1))$ bits of information about the players' inputs. Our result demonstrates that the naive protocol for computing $f^{\oplus n}$ is both information-theoretically optimal and asymptotically tight in error trade-offs.
\end{abstract}
\newpage

\tableofcontents

\newpage

\section{Introduction}
\label{sec:intro}

For a function $f: \mathcal{X} \rightarrow \{0,1\}$ and any natural number $n$, let $f^{\oplus n}:\mathcal{X}^n \rightarrow \{0,1\}$ denote the function defined by $f^{\oplus n}(X_1, \ldots, X_n) = f(X_1) \oplus \dots \oplus f(X_n)$. The questions surrounding \emph{XOR Lemmas} focus on the relationship between the resources needed to compute $f$ and those required for computing $f^{\oplus n}$. Given a procedure $\mathcal{P}$ for computing $f$, one could naively compute $f^{\oplus n}(X_1,\dots,X_n)$ by evaluating each $f(X_i)$ independently via $\mathcal{P}$ and then taking the XOR of the $n$ output bits. This strategy uses $n$ times the resources required to compute $f$. The central question of XOR lemmas asks whether this naive protocol is resource-optimal.

\begin{question}
   For which \emph{regimes of parameters} $(\rho, \rho')$ and which notions of \emph{resources} does computing $f^{\oplus n}$ with probability $\rho'$ require $\Omega(n)$ times the resources needed for computing $f$ with probability $\rho$?
\label{ques:main_question}
\end{question}

XOR lemmas are closely connected to the \emph{Direct Sum} problem, where we seek to compute $f(X_i, Y_i)$ for all $i \in [n]$, and have been extensively studied under various resource models, including circuit size \cite{Yao82, Lev87, Imp95, IW97, GNW11}, query complexity \cite{Sha03, She11, Dru12, BKLS20, BKST24}, and decision-tree complexity \cite{Hoz24}. Another related question is the \emph{Direct Product} problem: does the success probability of $f^{\oplus n}$ or the \emph{advantage} of $f^{\oplus n}$ decay exponentially as $n$ increases under limited resources? This question has been studied in contexts of game values~\cite{Raz98,Holenstein07,Raz08,Rao08}.

In this paper, we consider the computation of a function $f: \mathcal{X} \times \mathcal{Y} \rightarrow \{0,1\}$ in a \emph{two-player randomized communication} setting. In this model, Alice receives an input $X \in \mathcal{X}$, and Bob receives $Y \in \mathcal{Y}$. The goal is for the players to compute $f(X,Y)$ by exchanging a sequence of messages $\M = (M^1, M^2, \dots, M^r)$. For odd rounds $i$, Alice generates $M^i$ based on her input $X$ and the preceding messages $M^{<i}$; for even $i$, Bob generates $M^i$ based on $Y$ and the preceding messages $M^{<i}$. Both players also have access to private randomness and shared public randomness. The randomized messaging schemes $\M = (M^1, M^2, \dots, M^r)$ are called \emph{(randomized) protocols}. Notably, the computation of $f^{\oplus n}$ can also be modeled in this two-player setting: Alice receives $(X_1, \dots, X_n)$, Bob receives $(Y_1, \dots, Y_n)$, and their goal is to compute $f^{\oplus n}(X_1, \dots, X_n, Y_1, \dots, Y_n)$ through running a protocol. Many results exist when the resource of interest is the total length of messages, namely the \emph{communication complexity}, such as Direct Sum results \cite{CSWY01, Sha03, JRS03, HJMR07, BBCR10}, Direct Product results \cite{Kla10, BRWY13a, BRWY13b, Jai15, IR24b}, and XOR lemmas \cite{BBCR10, Yu22, IR24a, IR24b}.

Beyond suggesting trade-offs in the required resources, the naive protocol also provides insight into the optimal trade-offs between error rates, which occur when $(\rho, \rho') = \left(\frac{1}{2} + \frac{\alpha}{2}, \frac{1}{2} + \frac{\alpha^n}{2}\right)$ for some \emph{advantage} $\alpha \in (0,1)$. To see this, suppose $p$ and $q$ are independent $\{0,1\}$-valued random variables with advantages $\alpha_p$ and $\alpha_q$, meaning they can be predicted by $\hat{p}$ and $\hat{q}$ with probabilities $\frac{1}{2} + \frac{\alpha_p}{2}$ and $\frac{1}{2} + \frac{\alpha_q}{2}$, respectively. Then, the XOR of $p \oplus q$ can be predicted by $\hat{p} \oplus \hat{q}$ with probability $\frac{1}{2} + \frac{\alpha_p \alpha_q}{2}$, since
\begin{align*}
   \Pr(\hat{p} \oplus \hat{q} = p \oplus q) &= \Pr(\hat{p} = p \land \hat{q} = q) + \Pr(\hat{p} \ne p \land \hat{q} \ne q) \\
   &= \Pr(\hat{p} = p) \cdot \Pr(\hat{q} = q) + \Pr(\hat{p} \ne p) \cdot \Pr(\hat{q} \ne q) \tag{$p \perp q$} \\
   &= \left(\frac{1}{2} + \frac{\alpha_p}{2}\right) \cdot \left(\frac{1}{2} + \frac{\alpha_q}{2}\right) + \left(\frac{1}{2} - \frac{\alpha_p}{2}\right) \cdot \left(\frac{1}{2} - \frac{\alpha_q}{2}\right) \\
   &= \frac{1}{2} + \frac{\alpha_p \alpha_q}{2}.
\end{align*}
In the naive protocol, each $f(X_i, Y_i)$ is computed with probability $\rho = \frac{1}{2} + \frac{\alpha}{2}$, achieving advantage $\alpha$. Thus, computing the XOR of these $n$ bits yields an advantage of $\alpha^n$, corresponding to $\rho' = \frac{1}{2} + \frac{\alpha^n}{2}$.\footnote{It is reasonable to consider the regimes in terms of advantage $(\rho, \rho')$ or error probabilities $(1-\rho, 1-\rho')$ only \emph{asymptotically}, as it is possible to \emph{boost} the success probability by making multiple independent runs of the protocol, taking the majority answer. For example, $T = O(1)$ runs suffice to boost the success probability from $0.501$ to any constant below 1, or from error $n^{-0.01}$ to $n^{-O(1)}$.} Thus, Question \ref{ques:main_question} is worth investigating in two parameter regimes:
\begin{enumerate}
    \item $(\rho, \rho') = (\frac{9}{10}, \frac{1}{2} + 2^{-n})$, corresponding to the optimal trade-off where $\alpha = \Theta(1)$, and
    \item $(\rho, \rho') = (1-\frac{1}{n}, \frac{9}{10})$, corresponding to the optimal trade-off where $\alpha = 1 - \Theta(1/n)$.
\end{enumerate}

\subsection{Our Results}
In this paper, we provide an affirmative answer to Question \ref{ques:main_question} (up to vanishing additive losses) in the regime where $(\rho, \rho') = (1 - \frac{1}{n}, \frac{9}{10})$ when the resource of interest is \emph{information}. For a function $g$ and error parameter $\varepsilon \in (0,1)$, let the \emph{information complexity} of $g$ with error $\varepsilon$, denoted $\IC(g, \varepsilon)$, be the maximum amount of information each player learns about the other’s input by the end of a protocol that computes $g$ with probability at least $1 - \varepsilon$. This “resource” represents the (worst-case) amount of information players must learn to compute $f$ accurately. With this notion, we prove a strong XOR lemma for information complexity.

\begin{restatable}[Strong XOR Lemma for Information Complexity]{them}{xoric} 
    There exists a universal constant $\lambda \in (0,1)$ and $c_1 > 0$ such that for any function $f: \mathcal{X} \times \mathcal{Y} \rightarrow \{0,1\}$ and any positive integer $n$, we have
    $$\IC(f^{\oplus n}, 1/10) \geq c_1 n \cdot \left(\IC(f, n^{-1}) - \frac{\log \left(|\mathcal{X}| \cdot |\mathcal{Y}|\right)}{n^{\lambda}} - 1\right).$$
    \label{thm:xor_ic}
\end{restatable}

Our result is asymptotically tight (up to vanishing additive losses), as demonstrated by the following result, whose proof will be deferred to Appendix \ref{appendix:missing_proofs}.

\begin{restatable}{them}{xoricub}
    There exists a universal constant $c_2 > 0$ such that for any $\{0,1\}$-valued function $f$ and positive integer $n$, we have
    $$\IC(f^{\oplus n}, 1/10) \leq c_2 n \cdot \IC(f, n^{-1}).$$
    \label{thm:xor_ic_ub}
\end{restatable}

Our proof of Theorem \ref{thm:xor_ic} relies on a distributional version of the XOR lemma for \emph{information cost}. For a function $g$, error parameter $\varepsilon \in (0,1)$, and input distribution $\mu$, let $\IC_{\mu}(g, \varepsilon)$ denote the \emph{information cost} of $g$ with error $\varepsilon$, defined as the minimum information learned by each player about the other’s input while using a protocol that computes $g$ with probability at least $1 - \varepsilon$ \emph{over inputs drawn from $\mu$}. We establish a strong XOR lemma for distributional information cost, using it to prove Theorem \ref{thm:xor_ic}.

\begin{restatable}[Strong XOR Lemma for Distributional Information Cost]{them}{xordistic}
There exists a universal constant $\lambda \in (0,1)$ and $c_3 > 0$ such that for any function $f: \mathcal{X} \times \mathcal{Y} \rightarrow \{0,1\}$, any positive integer $n$, and any input distribution $\mu$ over $\mathcal{X} \times \mathcal{Y}$, we have
    $$\IC_{\mu^n}(f^{\oplus n}, 1/10) \geq c_3 n \cdot \left(\IC_{\mu}(f, n^{-\lambda}) - \frac{\log \left(|\mathcal{X}| \cdot |\mathcal{Y}|\right)}{n^{\lambda}} - 1\right).$$
    \label{thm:xor_dist_ic}
\end{restatable}

\subsection{XOR Lemmas for Exponentially Small Advantage}

Another significant regime is $(\rho, \rho') = (\frac{2}{3}, \frac{1}{2} + 2^{-n})$, representing the tight upper bound for $\alpha = \Theta(1)$. However, in this setting, an XOR lemma for information complexity does not hold. Consider the following protocol that computes $f^{\oplus n}(X_1,\dots,X_n, Y_1,\dots,Y_n)$ with probability $\frac{1}{2} + 2^{-n}$ while revealing only an exponentially small amount of information: for a $2^{-n+1}$ fraction of inputs, Alice sends her entire input to Bob, allowing him to compute $f^{\oplus n}$ exactly; otherwise, Alice sends nothing, and Bob outputs a random bit which is correct with probability $1/2$. This protocol achieves success probability $\frac{1}{2} + 2^{-n}$, but Alice reveals only $2^{-n+1} \cdot n \cdot \log |\mathcal{X}|$ bits of information, meaning that an XOR lemma cannot hold in this regime.

On the contrary, a recent result by \cite{Yu22} provides a positive answer to Question \ref{ques:main_question} when the resource of interest is \emph{communication}, showing that if any $r$-round protocol computing $f$ with probability $2/3$ requires $C$ bits of communication, then any protocol computing $f^{\oplus n}$ with probability $1/2 + 2^{-n}$ requires $n \cdot \left(r^{-O(r)} \cdot C - 1\right)$ bits. \cite{IR24b} extends this to the \emph{Direct Product} setting, as well as eliminating the exponent ``$-O(r)$" from Yu’s result.

\section{Related Work}
\label{sec:related_work}

We restrict our attention to the following question, as it immediately implies our Distributional XOR Lemma. 

\begin{question} 
Given a communication protocol $\pi$ for computing $f^{\oplus n}$ with error $1/10$ over $\mu^n$ with information cost $\mathcal{I}$, can we construct a new protocol $\eta$ for computing $f$ with error $n^{-O(1)}$ over $\mu$ with information cost $\approx \mathcal{I}/n$?
\end{question}

An easier variant of this question is implied by known results, where we allow the same constant error $1/10$ for computing $f$.

\subsection{The ``Folklore'' Input Embedding Procedure}
The XOR Lemma for distributional information cost is known be true in the setting where $\rho = \rho'$ due to the work by \cite{BYJKS04} which was made explicit by \cite{BBCR10}.

\begin{theorem}[Theorem 2.4 of \cite{BBCR10}; informal] For any function $f$, error parameter $\varepsilon \in (0,1)$, and an input distribution $\mu$, the following holds.

\textbf{Given} a protocol $\pi$ for computing $f^{\oplus n}$ with error $\varepsilon$ over $\mu^n$ with information cost at most $\mathcal{I}$, 

\textbf{then} there exists a protocol $\eta$ for computing $f$ with error $\varepsilon$ over $\mu$ with information cost $\mathcal{I}/n +O(1).$
\label{thm:embedding}
\end{theorem}

For completeness, we roughly sketch the proof of the theorem by constructing the protocol $\eta$ via by \emph{embedding} an input $(x,y)$ into one of the $n$ coordinates, and then execute $\pi$.

\begin{figure}[H]
    \centering
    \begin{tabular}{|p{15.5cm}|}
    \hline ~\\
    \multicolumn{1}{|c|}{\textbf{Protocols $\eta$ \ for computing $f(x,y)$ where $(x,y)$ are drawn from $\mu$.}} \\ 
    \\ \hline
    \begin{enumerate}
        \item Players use public randomness to sample a uniform index $J \in [n]$ and partial inputs $X_{<J}$ and $Y_{>J}$
        \item Alice embeds $X_J = x$ and privately samples $X_{>J}$ conditioned on $Y_{>J}$.
        \item Bob embeds $Y_J = y$ and privately samples $Y_{<J}$ conditioned on $X_{<J}$.
        \item Players execute $\pi$ to compute $f^{\oplus n}(X_1, \ldots, X_n, Y_1, \ldots, Y_n)$.
        \item Alice sends Bob an extra bit indicating $f^{\oplus n-J}(X_{>J}, Y_{>J})$, and Bob sends Alice an extra bit indicating $f^{\oplus J-1}(X_{<J}, Y_{<J})$.
        \item Players recover $f(x,y)$ by computing
        $$f(x,y) := f^{\oplus n}(X_1, \ldots, X_n, Y_1, \ldots, Y_n) \oplus f^{\oplus n-J}(X_{>J}, Y_{>J}) \oplus f^{\oplus J-1}(X_{<J}, Y_{<J}). $$
    \end{enumerate}
    \\
    \hline
    \end{tabular}
    \caption{A protocol $\eta$ for computing $f$ on inputs $(x,y) \sim \mu$ via the embedding method.}
    \label{table:decompose_protocols}
\end{figure}

The protocol $\eta$ can also be interpreted as follows: we list $n$ protocols $(\pi_1, \ldots, \pi_n)$ for which $\pi_j$ corresponds to $\eta \mid J = j$, and execute a $\pi_j$ for a random $j \in [n]$. It can be shown by calculation that the information costs of these $n$ protocols sum up to at most $\mathcal{I} + O(n)$. Since $\eta$ picks a uniform random coordinate $j$ and runs $\pi_j$, its information cost is at most 
$$ \frac{1}{n} \cdot (\mathcal{I} + O(n)) = \frac{\mathcal{I}}{n} + O(1). $$

Nevertheless, $\eta$ is ineffective in boosting correctness as it only succeeds with probability $\rho' = \rho$. To see this, observe that both players always compute $f^{\oplus n-J}(X_{>J}, Y_{>J})$ and $f^{\oplus J-1}(X_{<J}, Y_{<J})$ correctly. Thus, the correctness of $f(x,y)$ inherits from that of $f^{\oplus n}(X_1, \ldots, X_n, Y_1, \ldots, Y_n)$. Since $(X_1, \ldots, X_n, Y_1, \ldots, Y_n)$ is distributed exactly like $\mu^n$, the probability that $\eta$ is correct remains $\rho$. This roughly proves Theorem \ref{thm:embedding}.

By plugging $\varepsilon = 1/10$ (i.e. the success rate is $\rho = \rho' = 9/10$) into Theorem~\ref{thm:embedding}, and take the supremum over protocol $\pi$, we have shown a distributional XOR lemma with preserving error. As a corollary, a similar bound can be shown for the information complexity.\footnote{We shall not show the reduction explicitly; however, it is remarkably similar to the proof of Theorem \ref{thm:xor_ic_ub} appeared in Appendix \ref{appendix:missing_proofs}.}

\begin{theorem} There exists a universal constant $c_4$ such that the following holds. For any $\{0,1\}$-valued function $f$, any positive integer $n$, and any input distribution $\mu$, it holds that 
    $$\IC_{\mu^n}(f^{\oplus n}, 1/10) \geq n \cdot \left(\IC_{\mu}(f, 1/10) - O(1)\right)$$
and consequently
$$\IC(f^{\oplus n}, 1/10) \geq c_4n \cdot \left(\IC(f, 1/10) - O(1)\right).$$
\label{thm:dist_err_preserving}
\end{theorem}

Our main results in Theorem~\ref{thm:xor_ic} and Theorem~\ref{thm:xor_dist_ic} can be interpreted as improving the asymptotic of errors of Theorem~\ref{thm:dist_err_preserving} from $(\frac{1}{10}, \frac{1}{10})$ to $(\frac{1}{10}, n^{-O(1)})$ where it is asymptotically-tight. Attempting to obtain polynomially-small error for $f$ poses as the main technical challenge of our work.

\subsection{Driving Down the Errors}

To bring the error probability of $f$ down from $\frac{1}{10}$ to $n^{-1}$ in the XOR Lemmas, we adopt an alternative view of the input embedding procedure, as observed by \cite{Yu22}. On a high level, Yu proposes a \emph{decomposition} procedure that split the protocol $\pi$ for computing $f^{\oplus n}$ over $\mu^n$ into two protocols: a protocol $\pi^{(n)}$ for computing $f$ over $\mu$, and  a protocol $\pi^{(<n)}$ for computing $f^{\oplus n-1}$ over $\mu^{n-1}$, for which their information costs add up to $\mathcal{I} + O(1)$. More importantly, it holds \emph{pointwisely} that the advantage of $\pi$ is equal to the product of the advantage of $\pi^{(n)}$ and the advantage of $\pi^{(<n)}$. These observations motivate the reasoning that at least one of the following cases should occur:
\begin{enumerate}
    \item[(1)] $\pi^{(n)}$ is a ``good'' protocol for computing $f$, as it has low information cost and errs with small probability. In this case, we have found the desired protocol $\eta := \pi^{(n)}$.
    \item[(2)] $\pi^{(<n)}$ is ``better than average'' for computing $f^{\oplus n-1}$. In this case, we recursively decompose $\pi^{(<n)}$ into $\pi^{(n-1)}$ and $\pi^{(<n-1)}$, until we land in case (1).
\end{enumerate}

This preliminary idea of \emph{protocol decomposition} led \cite{Yu22} to the Strong XOR Lemma for \emph{communication complexity} in the regime where $(\rho, \rho') = \left(\frac{1}{2} + 2^{-n}, \frac{2}{3}\right)$. However, to the best of our knowledge, the majority of the techniques used in in Yu's proof do \emph{not} transfer to our regime where $(\rho, \rho') = \left(\frac{9}{10}, 1 - \frac{1}{n}\right)$. While Yu’s work serves as a foundational building block, our approach quickly diverges. We address some of these distinctions in Section \ref{subsec:distinction}.

\section{Technical Overview}
\label{sec:tech_overview}

In this section, we present our main lemma and outline its proof. Its complete proof can be found in Section \ref{sec:decomposition} and Section \ref{sec:standardize}. Our main results (Theorem \ref{thm:xor_ic} and Theorem \ref{thm:xor_dist_ic}) follow directly from the main lemma. For clarity, their proofs are deferred to Section \ref{sec:pfs_xor}.

\begin{restatable}[Main Technical Lemma]{lem}{mainlemma}
There exists a universal constant $C > 0$ and $\lambda \in (0,1)$ such that for any function $f: \mathcal{X} \times \mathcal{Y} \rightarrow \{0,1\}$, the following statement holds.

\textbf{If} there exists a (standard) communication protocol $\pi$ for computing $f^{\oplus n}$ over an input distribution $\mu^n$ such that it errs with probability $\frac{1}{10}$ and has information cost $\mathcal{I}$,

\textbf{then} there exists a (standard) communication protocol $\eta$ for computing $f$ over an input distribution $\mu$ such that it errs with probability $n^{-\lambda}$ and has information cost at most 
$C \cdot \left(\frac{\mathcal{I}}{n} + \frac{\log \left(|\mathcal{X}| \cdot |\mathcal{Y}|\right)}{n^{\lambda}} + 1\right).
$
\label{lem:main_lemma}
\end{restatable}

To prove Lemma \ref{lem:main_lemma}, we assume its setup. Let $\pi$ be a protocol that computes $f^{\oplus n}$ on the input distribution $\mu^n$ with information cost $\mathcal{I}$ and advantage $\frac{4}{5}$.\footnote{The advantage $\frac{4}{5}$ corresponds to the error probability $\frac{1}{10}$.} Below, we sketch the approach to obtaining the protocol $\eta$ that satisfies the requirements of Lemma \ref{lem:main_lemma}.

\subsection{Binary Protocol Decomposition from \texorpdfstring{\cite{Yu22}}{Yu22}}
Our starting point is a slight modification of the protocol decomposition procedure introduced in \cite{Yu22}. It is worth noting that in the work of \cite{Yu22}, their decomposition yields two \emph{unbalanced} protocols: a protocol $\pi^{(n)}$ for computing $f$, and a protocol $\pi^{(<n)}$ for computing $f^{\oplus n-1}$. In our work, we split the protocol equally so the two smaller protocols both compute $f^{\oplus n/2}$. This turns out to be an important aspect of our decomposition, as it will later yields a fairly clean analysis. Specifically, given a protocol $\pi$ for computing $f^{\oplus n}$, we construct two protocols, namely $\pi_0$ and $\pi_1$, as follows.

\begin{algorithm}[H]
\caption{: Protocol $\pi_0$ for computing $f^{\oplus n/2}$  \\ \textbf{Input:} Alice receives input $x$ on $n/2$ coordinates and Bob receives input $y$ on $n/2$ coordinates}
\begin{algorithmic}[1]
\State{Alice sets $X_{\leq n/2}$ to $x$ and Bob sets $Y_{\leq n/2}$ to $y$}
\State{Alice and Bob publicly samples $Y_{>n/2}$}
\State{Alice privately samples $X_{>n/2}$}
\State{Players run $\pi$ pretending that their inputs are $(X,Y)$ on $n$ instances to compute $f^{\oplus n}(X,Y)$}
\State{Alice sends Bob an extra bit $b_0$ indicating $f^{\oplus n/2}(X_{>n/2}, Y_{>n/2})$}
\State{Both players recover $f^{\oplus n/2}(x,y) = f^{\oplus n/2}(X_{\leq n/2}, Y_{\leq n/2}) = f^{\oplus n}(X,Y) \oplus b_0.$}
\end{algorithmic}
\end{algorithm}

\begin{algorithm}[H]
\caption{: Protocol $\pi_1$ for computing $f^{\oplus n/2}$ \\ \textbf{Input:} Alice receives input $x$ on $n/2$ coordinates and Bob receives input $y$ on $n/2$ coordinates}
\begin{algorithmic}[1]
\State{Alice sets $X_{>n/2}$ to $x$ and Bob sets $Y_{>n/2}$ to $y$}
\State{Alice and Bob publicly samples $X_{\leq n/2}$}
\State{Bob privately samples $Y_{\leq n/2}$}
\State{Players run $\pi$ pretending that their inputs are $(X,Y)$ on $n$ instances to compute $f^{\oplus n}(X,Y)$}
\State{Bob sends Alice an extra bit $b_1$ indicating $f^{\oplus n/2}(X_{\leq n/2}, Y_{\leq n/2})$}
\State{Both players recover $f^{\oplus n/2}(x,y) = f^{\oplus n/2}(X_{>n/2}, Y_{>n/2}) = f^{\oplus n}(X,Y) \oplus b_1.$}
\end{algorithmic}
\end{algorithm}

One might notice that the protocol $\pi_0$ could still achieve the same success probability even if the \emph{final bit} $b_0$ were omitted. This is because, by computing $b_0$ herself, Alice can answer $f^{\oplus n/2}(X_0, Y_0)$ on Bob’s behalf. However, the final bit remains essential in our decomposition, as we will later require \emph{both players} to learn the value of $b_0$. As a result, $b_0$ could add up to one bit to the information cost of $\pi_0$. In what follows, we make one simplification: when analyzing the information costs, we account only for the cost incurred by the messages in the original protocol $\pi$, neglecting the cost of the final bits. Eventually, we will address how to lift this assumption.

Observe that in both protocols $\pi_0$ and $\pi_1$, the input pair  $(x,y)$ is drawn from the distribution $\mu^{n/2}$. Moreover, to be able to execute $\pi$, the players pretend that their inputs are consisting of $n$ instances by filling up the missing coordinates so that their ``artificial'' inputs $(X,Y)$ distribute exactly like $\mu^n$. Notice further that in $\pi_0$, Alice knows both $X_{>n/2}$ and $Y_{>n/2}$; thus, she computes $b_0$ correctly with probability 1. Similarly, in $\pi_1$ Bob computes $b_1$ correctly with probability 1.

Let $I_0$ and $I_1$ denote the information costs of $\pi_0$ and $\pi_1$ respectively. \cite{Yu22} observed that when decomposing the protocols as above, their information costs and advantages also admits algebraic decomposition.

\paragraph{Decomposition of Information Costs.} For $\pi_0$, the information cost from Alice's side is $I(\M: Y_{\leq n/2} \mid X_{\leq n/2}Y_{>n/2})$, and from Bob's side is $I(\M: X_{\leq n/2} \mid Y)$. For $\pi_1$, the information cost from Alice's side is $I(\M: Y_{>n/2} \mid X)$, and from Bob's side is $I(\M: X_{>n/2} \mid X_{\leq n/2}Y_{>n/2})$. Observe that 

\begin{align*}
& I(\M: Y_{\leq n/2} \mid X_{\leq n/2} Y_{>n/2} ) + I(\M: Y_{>n/2} \mid X) \\ & = I(\M: Y_{\leq n/2} \mid X Y_{>n/2}) + I(\M: Y_{>n/2} \mid X) \\
& = I(\M:Y \mid X) \tag{chain rule of information}
\end{align*}
where the first equality follows from the \emph{rectangle property} of communication protocols.\footnote{The rectangle property will be made explicit in Section \ref{sec:formalize_protocols}.} Similarly we also have $I(\M: X_{\leq n/2} \mid Y_{\leq n/2} X_{>n/2} ) + I(\M: X_{>n/2} \mid Y) = I(\M:Y \mid X)$. Therefore, we have $I_0 + I_1 = \mathcal{I}$. This suggests that the information costs of the protocols decompose additively.

\paragraph{Decomposition of Advantage.} Denote the following set of random variables which depends on the randomness of $\M X_{\leq n/2}Y_{>n/2}$:
\begin{align*} A_0 & = \adv(f^{\oplus n/2}(X_{\leq n/2}, Y_{\leq n/2}) \mid \M X_{\leq n/2}Y_{>n/2}) \\
A_1 & = \adv(f^{\oplus n/2}(X_{>n/2}, Y_{>n/2}) \mid \M X_{\leq n/2}Y_{>n/2}) \\
Z & = \adv(f^{\oplus n}(X,Y) \mid \M X_{\leq n/2}Y_{>n/2}) 
\end{align*}
where $\adv(a \mid W)$ denotes the advantage of the bit $a$ conditioned on event $W$.

Conditioned on $\M X_{\leq n/2}Y_{>n/2}$, $A_0$ is the advantage of $\pi_0$ (from Alice's perspective) and $A_1$ is the advantage of $\pi_1$ (from Bob's perspective). Moreover, we have $f^{\oplus n/2}(X_{\leq n/2}, Y_{\leq n/2}) \perp f^{\oplus n/2}(X_{>n/2}, Y_{>n/2}) \mid \M X_{\leq n/2}Y_{>n/2}$ due to the \emph{rectangle property} of communication protocols. Therefore, we have $Z = A_0A_1$.\footnote{If $p,q$ are independent $\{0,1\}$-valued random variables, then $\adv(p \oplus q) = \adv(p) \cdot \adv(q)$.} This suggests that advantages of the protocols decompose multiplicatively.

However, there is a catch: the equality $Z = A_0 A_1$ only holds \emph{pointwise}. Yet we know that the advantage of $\pi_0$ is $\E(A_0)$, the advantage of $\pi_1$ is $\E(A_1)$, and the advantage of $\pi$ is at most $\E(Z)$. Only did we have an assumption that $\E(A_0)\E(A_1) \geq \E(A_0A_1) = \E(Z)$, we could have concluded that advantages decompose multiplicatively. However, the reality is that it is not always the case that the $``\geq''$ holds.

As a thought experiment, let us examine the (incorrect) implications of the information costs and advantages decompositions if they had held. If we apply the decomposition on $\pi_0$ and $\pi_1$ again, we could obtain four protocols for computing $f^{\oplus n/4}$. Repeating this recursively, we would eventually end up with $n$ protocols for computing $f$. By the additive decomposition of information costs and the (incorrect) multiplicative decomposition of advantages, those $n$ protocols must have their information costs summing to $\mathcal{I}$ and their advantages multiplying to at least $4/5$. By the averaging argument, there must exist a protocol $\eta$ computing $f$ with information cost $O(\mathcal{I}/n)$ and advantage at least $(4/5)^{1/n} = 1- \Theta(1/n)$ which is equivalent to success probability $1-\Theta(1/n)$. 

While this approach fails due to the false premise, the idea of recursively breaking down the protocols remains useful.

\subsection{Our Approach: ``Conditional'' Protocol Decomposition}

To address the aforementioned issues, we propose a modified protocol decomposition approach. Specifically, let $W$ denote the event that $Z \geq 0.01$. Then, we obtain protocols $\pi_0$ and $\pi_1$ by applying binary decomposition to the ``protocol'' $\pi \mid W$ (thus the name \emph{conditional} decomposition).\footnote{Careful readers might flag that $\pi \mid W$ no longer aligns with the conventional description of communication protocols. We will address this point shortly.}

We also introduce a new parameter of interest, called \emph{disadvantage}, denoted $\varepsilon$, which is defined as one minus the protocol’s advantage. Generally, this quantity is twice of the error probability when outputting the more-likely bit. Thus, by having players output the more-likely bit, the disadvantage provide a measure of error probability, accurate within a factor of 2. Let $\varepsilon = 1 - \frac{4}{5} = \frac{1}{5}$ denote the disadvantage of $\pi$, and let $\varepsilon_0$ and $\varepsilon_1$ represent the disadvantages of $\pi_0$ and $\pi_1$, respectively.

\begin{figure}[H]
    \centering

\begin{center}
\begin{tikzpicture}[node distance=2cm]

\tikzstyle{protocol} = [rectangle, 
minimum width=1cm, 
minimum height=1cm, 
text centered, 
text width=5cm, 
draw=black, 
fill=orange!30]

\tikzstyle{arrow} = [thick,->,>=stealth]

\node (pi) [protocol, yshift=-0.5cm, xshift=-0.5cm] { $\pi$ \\ \begin{enumerate}
    \item compute $f^{\oplus n}$
    \item information cost $\mathcal{I}$
    \item disadvantage $\varepsilon$
\end{enumerate}};

\node (pi_w) [protocol, below of=pi, yshift=-1cm] { $\pi \mid W$};
\draw [arrow] (pi) -- node[anchor=west] {Conditioning event $W$} (pi_w) ;

\node (pi0) [protocol, below of=pi_w, yshift=-1cm, xshift = -4cm] { $\pi_0$
\begin{enumerate}
    \item compute $f^{\oplus n/2}$
    \item information cost $I_0$
    \item disadvantage $\varepsilon_0$
\end{enumerate}};
\draw [arrow] (pi_w) -- (pi0);

\node (pi1) [protocol, below of=pi_w, yshift=-1cm, xshift = 4cm] { $\pi_1$
\begin{enumerate}
    \item compute $f^{\oplus n/2}$
    \item information cost $I_1$
    \item disadvantage $\varepsilon_1$
\end{enumerate}};
\draw [arrow] (pi_w) -- (pi1);

\end{tikzpicture}
\end{center}
\caption{A single-level ``conditional'' decomposition of a protocol $\pi$ for $f^{\oplus n}$ into two protocols $\pi_0$ and $\pi_1$ for $f^{\oplus n/2}$.}
\end{figure}

With the new way of decomposition, we can show that \footnote{In fact, we will prove stronger statements in Section \ref{sec:decomposition}.}
\begin{equation}
    \varepsilon_0 + \varepsilon_1 \leq 1.98 \varepsilon \hspace{5mm} \text{ and } \hspace{5mm} I_0 + I_1 \leq \mathcal{I} \cdot e^{O(\varepsilon)}.
\label{eq:one-level-decomp}
\end{equation}
In other words, the conditional decomposition yields a \emph{geometric} decay in disadvantages and a \emph{near-linear} decay in information costs, adjusted by a small multiplicative factor of $e^{O(\varepsilon)}$.

Next, consider applying this procedure recursively. For each protocol $\pi_S$, we introduce a conditioning event $W_S$ and derive two smaller protocols, $\pi_{S0}$ and $\pi_{S1}$, via a binary decomposition of $\pi_S \mid W_S$. Let the collection of protocols $\pi_S$ with $|S| = k$ be referred to as the \emph{level $k$}. Upon reaching level $m = \log_2{n}$, we have obtained $n$ protocols $\{\pi_S\}_{|S| = m}$, each of which computes $f$.

\paragraph{Decomposition of Disadvantages.} By the first inequality of (\ref{eq:one-level-decomp}), we see that the average of $\varepsilon$'s across level $k+1$ decreases from that of level $k$ to a multiplicative factor of $0.99$. As a consequence, the average $\varepsilon$'s for the leaf level $m = \log_2{n}$ (i.e. where the protocols are for single-instance of $f$) is at most $(0.99)^{\log_2{n}} \varepsilon < n^{-0.01}$.

\paragraph{Decomposition of Information Costs.} The breakdown of information costs is much more intricate. To illustrate potential outcomes following Equation (\ref{eq:one-level-decomp}), let us assume that the $\varepsilon$ values are well-balanced: for $S$, we consider $\varepsilon_{S0} \approx \varepsilon_{S1} \approx 0.99 \varepsilon_S$. This implies that $\varepsilon_S \approx (0.99)^{|S|} \varepsilon$. Under this assumption, the \emph{total} information costs of the protocols at level $k+1$ increase from those at level $k$ by a factor of $\exp(0.99^k \varepsilon)$. Overall, this accumulation leads to a blow-up factor of $O(1)$ thanks to the geometric sum. Therefore, we can expect the sum of information costs at level $m = \log_2{n}$ to be $O(\mathcal{I})$.

By averaging arguments, there must exist a protocol $\eta$ at the leaf level $m = \log_2{n}$ with an information cost of $O(\mathcal{I}/n)$ and an error probability of $n^{-0.01}$.

Nonetheless, the argument outlined above encounters several technical difficulties.

\begin{enumerate}
    \item[(1)] The calculation assumes a balanced split of information costs across all decompositions; however, this cannot be guaranteed.\footnote{One might be tempted to use the average of $\varepsilon$'s across level $k$ as a proxy to the exponent of the blowups from level $k$ to level $k+1$, but to the best of our attempt, this approach fails due to the convexity of $e^x$.}
\end{enumerate}

Our earlier argument can be interpreted as a probabilistic proof: we sample an index $S$ uniformly from $[n]$, and in expectation, $\pi_S$ has small disadvantage and low information cost. To address issue (1), we instead sample $S$ from a carefully-constructed distribution $\mathcal{D}$ where $\D(S)$ is \emph{proportional} to the probability of all conditioning events attached to $\pi_S$ along the recursive decomposition. We will then show that in expectation over $S \sim \D$, a protocol $\pi_S$ exhibits the desired properties: it has small disadvantage of $n^{-0.01}$ and an information cost of approximately $O(\mathcal{I}/n)$. 

Recall that under our assumption, the information cost of $\pi_S$ does not account for the missing bits, denoted $B_S$, which include all the ``final bits'' along the recursive decomposition. Let us now address this assumption by arguing about the information cost incurred by $B_S$. Since each level of the decomposition adds one extra bit to the protocol, and $\pi_S$ is at level $m = \log_2{n}$, the length of $B_S$ must be of $\log_2 n$ bits. Naively, this means $B_S$ could contain up to $\log_2 n$ bits of information, which is too costly. To correct this, it can be shown that $B_S$ contains, in expectation, only $O(1)$ information about the inputs $(X_S, Y_S)$ at coordinate $S$. Therefore, the ``true'' information cost of the protocol $\pi_S$ is bounded by $O\left(\frac{\mathcal{I}}{n} + 1\right)$.

However, we are not finished yet; we still encounter additional issues:

\begin{itemize}
    \item[(2a)] $\pi_S$ is not a \emph{standard} communication protocol.
    \item[(2b)] The input distribution of $\pi_S$ is no longer $\mu$.
\end{itemize}

Let us first elaborate on issue (2a). In \emph{standard} communication protocols, each player is to generate the next message based solely on their input and the past messages exchanged (as well as their private and public randomness.) For instance, the protocol $\pi$ is standard due to the set-up of Lemma \ref{lem:main_lemma}. In contrast, $\pi \mid W$ does not meet this criterion because the conditioning event $W$ can introduce arbitrary correlations between Alice's messages and Bob's input, and vice versa. Protocols that allow such correlations are referred to as \emph{generalized} communication protocols. As a result, $\pi_S$ is no longer a standard protocol, yet it remains a generalized protocol.

Issue (2b) stems from a similar source: the sequence of conditioning events, denoted by $E$, that are recursively applied throughout the decomposition. This conditioning can distort the input distributions of $\pi_S$ from $\mu$ to $\mu \mid E$ in unpredictable ways.

Fortunately, both issues can be resolved simultaneously. The underlying intuition is that on average, each conditioning event occurs with a probability very close to 1, suggesting that the overall distortion induced by $E$ is minimal in expectation. Consequently, we can expect the protocol $\pi_S$ to be ``close'' to a standard protocol while operating on an input distribution that is ``close'' to $\mu$.

Somewhat-more formally, we can augment the protocol $\pi_S$ with another desirable set of properties: it is statistically-close to some standard protocol $\eta$ (in terms of KL-Divergence) whose input distribution is precisely $\mu$. Importantly, the protocol $\eta$ maintains an error probability of $n^{-0.001}$ over input distribution $\mu$ and achieves information cost of at most $C \cdot \left(\frac{\mathcal{I}}{n} + o_n(1) + 1\right)$ for some absolute constant $C$. In summary, the \emph{standard} protocol $\eta$ exhibits all the required properties, thereby proving Lemma \ref{lem:main_lemma}.

Notably, due to technical reasons, the final description of our the conditional decomposition procedure must deviate from the overview provided here. Nevertheless, the overall flow and main ideas of the proof remain largely the same.

\subsection{Key Differences from \texorpdfstring{\cite{Yu22}}{Yu22}}
\label{subsec:distinction}

As briefly mentioned, to the best of our understanding, the techniques from Yu's work do not extend the distributional XOR lemma to our regime, where $(\rho, \rho') = \left(\frac{9}{10}, 1 - \frac{1}{n}\right)$. At a glance, both our approach and Yu's adopt a similar strategy: recursively applying (conditional) decomposition until we obtain a protocol $\eta$ for computing $f$ with low ``cost'' and small distributional error. However, the sequence of conditioning events $E$ inevitably affects the distribution, distorting it in an unpredictable way. The key distinction lies in the specification of the conditioning events used in each paper.

To expand on this, at each level of decomposition, Yu’s approach involves a conditioning event that occurs with constant probability $O(1)$. Accumulating across all levels, $E$ occurs with probability $O(1)$ on average. This poses a fatal challenge in our regime, where we can tolerate only polynomially-small distributional error for $f$. To understand why, we recognize that the guarantee of “small distributional error'' of $\eta$ is evaluated against its own the input distribution $\mu \mid E$ , rather than the ``true'' input distribution $\mu$. Consequently, the error of $\eta$ with respect to $\mu$ is:
\begin{align*}
    & \Pr_{(x,y) \sim \mu}\left(\eta \text{ errs on }(x,y)\right) \\
    & = \Pr(E) \cdot \Pr_{(x,y) \sim \mu \mid E}\left(\eta \text{ errs on }(x,y) \right) + \Pr(\overline{E}) \cdot \Pr_{(x,y) \sim \mu \mid \overline{E}}\left(\eta \text{ errs on }(x,y) \mid E\right) \\
    & \geq \Pr(\overline{E}) \cdot \Pr_{(x,y) \sim \mu \mid \overline{E}}\left(\eta \text{ errs on }(x,y)\right)
\end{align*}
Since we have no guerantees over $\Pr_{(x,y) \sim \mu \mid \overline{E}}\left(\eta \text{ errs on } (x,y)\right)$, this probability can be as large as $\Omega(1)$, causing such error to be as large as $\Pr(\overline{E}) = \Omega(1)$.

In contrast, in our work, we propose a set of ``simple'' conditioning events, each of which occurs with probability $1 - o(1)$. These events result in $\Pr(E) = 1 - \frac{1}{\text{poly}(n)}$ on average. In this scenario, the error of $\eta$ with respect to the input distribution $\mu$ is polynomially-bounded:
\begin{align*}
    & \Pr_{(x,y) \sim \mu}\left(\eta \text{ errs on }(x,y)\right) \\
    & = \Pr(E) \cdot \Pr_{(x,y) \sim \mu \mid E}\left(\eta \text{ errs on }(x,y) \right) + \Pr(\overline{E}) \cdot \Pr_{(x,y) \sim \mu \mid \overline{E}}\left(\eta \text{ errs on }(x,y) \mid E\right) \\
    & \leq \Pr_{(x,y) \sim \mu \mid E}\left(\eta \text{ errs on }(x,y)\right) + \Pr(\overline{E}) \\
    & \leq \frac{1}{\text{poly}(n)}  \tag{$\eta$ errs w.p. $\frac{1}{\text{poly}(n)}$ on $\mu \mid E$, and $\Pr(\overline{E}) = \frac{1}{\text{poly}(n)}$}
\end{align*}
which within the desired range of distributional error. This rough calculation plays an important role in addressing issue (2b) in Section \ref{sec:standardize}.

In short, the amount of distortion in Yu's decomposition is too large to manage in the regime where we allow only polynomially small error for computing $f$, whereas our decomposition incurs only low distortion that remains manageable.

Another distinction between our approach and Yu's is that, at each level of the recursive procedure, we split a protocol for $f^{\oplus k}$ into two protocols for computing $f^{\oplus k/2}$, whereas Yu’s approach splits it into two protocols: one for computing $f$ and one for computing $f^{\oplus (k-1)}$. Such a ``binary'' decomposition is necessary to ensure the intersection of all events we condition on across the levels has high probability.

\subsection{Paper Organization}
In Section \ref{sec:prelim}, we establish the conventions used throughout this paper and review basic information theory concepts. Section \ref{sec:formalize_protocols} introduces key notations and properties of communication protocols, along with various cost functions relevant to our proofs. In Section \ref{sec:lemmas}, we prove essential lemmas. Section \ref{sec:decomposition} presents our recursive decomposition and addresses issue (1). In Section \ref{sec:standardize}, we tackle issues (2a) and (2b) simultaneously, thereby completing the proof of Lemma \ref{lem:main_lemma}. Finally, in Section \ref{sec:pfs_xor}, we use Lemma \ref{lem:main_lemma} to prove our main theorems. Note that some algebraically intensive proofs are deferred to Appendix \ref{appendix:missing_proofs}.

\section{Preliminaries}
\label{sec:prelim}

For variables $X$, $Y$, and $Z$, we write $X \perp Y \mid Z$ to indicate that $X$ and $Y$ are independent when conditioned on $Z$. Similarly, for an event $E$, we write $X \perp Y \mid E$ to indicate that $X$ and $Y$ are independent given that $E$ occurs. We may write the joint distributions interchangeably by $X,Y$ or $XY$.

Following standard notation, an uppercase letter represents a variable, while the corresponding lowercase letter denotes its value. For a distribution $\pi$ supported over multiple variables, let $\pi(X)$ represent the marginal distribution of $X$. For a value $x$, let $\pi(x)$ denote the probability that $X = x$ under $\pi$. For an event $W$, let $\pi(W)$ denote the probability of $W$ occurring under the distribution $\pi$. We also define the conditional analogues: let $\pi(X \mid Y)$ be the distribution of $X\mid Y$, let $\pi(X \mid y)$ represent the distribution of $X$ conditioned on $Y = y$, and let $\pi(X \mid W)$ denote the distribution of $X$ conditioned on event $W$ occurring. We write $X \sim \pi$ and $X \sim \pi \mid W$ to indicate sampling $X$ from the distribution $\pi(X)$ and the conditional distribution $\pi(X \mid W)$, respectively.

For a function $f: \mathcal{X} \times \mathcal{Y} \rightarrow \{0,1\}$ and any natural number $n$, we denote by $f^{\oplus n} : \mathcal{X}^n \times \mathcal{Y}^n \rightarrow \{0,1\}$ the function such that
$$
f^{\oplus n}(X_1, \dots, X_n, Y_1, \dots, Y_n) = f(X_1, Y_1) \oplus \dots \oplus f(X_n, Y_n).
$$
Occasionally, we may refer to $f^{\oplus n}$ simply as $f^{\oplus}$, omitting the number of XORed instances.

We also define \emph{advantage} and \emph{disadvantage} of $\{0,1\}$-valued random variables.

\begin{definition}[Advantage and Disadvantage.]
Let $\pi$ be a distribution and $b$ be a $\{0,1\}$-valued random variable. Denote the \emph{advantage} of $b$ with respect to $\pi$ to be $$\adv^{\pi}(b) := |2 \cdot \pi(b = 0) - 1| = |2 \cdot \pi(b = 1) - 1|.$$ 
Moreover, for any event $W$, we might write $\adv^{\pi}(b \mid W)$ and $\adv^{\pi \mid W}(b)$ interchangeably. Conversely, denote the \emph{disadvantage} of $b$ with respect to $\pi$ to be $1-\adv^{\pi}(b).$
\end{definition}

Note that disadvantage of $b$ is in fact twice the error probability of predicting $b$ with its more-likely value among $\{0,1\}$. To see this, suppose that $b$ takes on the value $0$ with probability $\frac{1}{2} + \frac{\alpha}{2}$. Then, by predicting $b$ with $0$ (i.e. its more-likely bit), we err with probability $\frac{1}{2} - \frac{\alpha}{2}$, which is half of $b$'s disadvantage.

\begin{fact} Let $\pi$ be a distribution, and let $b_1$ and $b_2$ be independent $\{0,1\}$-valued random variables in the same probability space $\pi$. Then, we have
$$\adv^{\pi}(b_1 \oplus b_2) =\adv^{\pi}(b_1) \cdot \adv^{\pi}(b_2).$$
\label{fact:xor_ind_bits}
\end{fact}

\begin{proof}
    Let $\Pr(b_1 = 0) = \frac{1}{2} + \frac{\beta_1}{2}$ and $\Pr(b_2 = 0) = \frac{1}{2} + \frac{\beta_2}{2}$. This means the advantage of $b_1$ and $b_2$ are $|\beta_1|$ and $|\beta_2|$ respectively. Then, 
    \begin{align*} \Pr(b_1 \oplus b_2 = 0) & = \Pr(b_1 = 0 \land b_2 = 0) + \Pr(b_1 = 1 \land b_2 = 1) \\
    & = \Pr(b_1 = 0) \cdot \Pr(b_2 = 0) + \Pr(b_1 = 1) \cdot \Pr(b_2 = 1) \tag{$b_1 \perp b_2$}\\
    & = \left(\frac{1}{2} + \frac{\beta_1}{2}\right) \cdot \left(\frac{1}{2} + \frac{\beta_2}{2}\right) +  \left(\frac{1}{2} - \frac{\beta_1}{2}\right) \cdot \left(\frac{1}{2} - \frac{\beta_2}{2}\right) \\
    & = \frac{1}{2} + \frac{\beta_1\beta_2}{2}
    \end{align*}
    meaning that the advantage of $b_1 \oplus b_2$ is $|\beta_1\beta_2|$ which is exactly the product of the advantages of $b_1$ and $b_2$.
\end{proof}

\subsection{Basic Information Theory}
\label{subsec:info_theory}

For the following set of definitions, let $X,Y,Z$ be arbitrary discrete variables in the probability space $\pi$.

\begin{definition}[Entropy and Conditional Entropy] We denote:
\begin{itemize}
    \item The \emph{entropy} of $X$ is defined as:
$$H(X) = \E_{x \sim \pi} \log \frac{1}{\pi(x)} = \sum_{x} \pi(x) \cdot \log \frac{1}{\pi(x)}.$$
We may abuse the notion that $0 \cdot \log \frac{1}{0} = 0$, or equivalently only consider the summation over $x \in supp(X)$.
    \item The \emph{conditional entropy} of $X \mid Y$ is defined as:
$$H(X \mid Y) = \E_{y \sim \pi} H(X \mid Y = y)$$
\end{itemize}
\end{definition}

\begin{theorem}[Chain Rule for Entropy] $H(X, Y) = H(X) + H(Y \mid X) = H(Y) + H(X \mid Y).$
\end{theorem}

Rearrange it, we have the definition of mutual information.

\begin{definition}[Mutual Information] The \emph{mutual information} between $X$ and $Y$ is defined as:
    $$I(X:Y) = H(X) - H(X \mid Y) = H(Y) - H(Y \mid X) = H(X) + H(Y) - H(X,Y).$$
\end{definition}

As a by-product, the entropy is subaddtive.

\begin{theorem}[Subadditivity of Entropy] It holds that 
$$H(X,Y) \leq H(X) + H(Y).$$
The equality is achieved when $X \perp Y$.
\end{theorem}

\begin{definition}[Conditional Mutual Information and Chain Rule]
The \emph{conditional mutual information} is defined as:
$$I(Y:Z \mid X) = I(XY : Z) - I(X:Z).$$
Rearranging it yields a \emph{chain rule} for mutual information:
$$I(XY : Z) =  I(X:Z) + I(Y:Z \mid X).$$
\end{definition}

The following definition measures the closeness of distributions.

\begin{definition}[KL-Divergence and Total Variation Distance] Let $\pi(X)$ and $\eta(X)$ distributions over the variable $X$. Denote the following distances between the two distributions.
\begin{enumerate}
    \item \emph{KL-Divergence}: $\dkl{\pi(X)}{\eta(X)} = \E_{x \sim \pi(X)} \log \frac{\pi(x)}{\eta(x)}$
    \item \emph{Total Variation Distance}: $\|\pi(X) - \eta(X)\| = \sum_{x} | \pi(x) - \eta(x)|$.
\end{enumerate} 
\end{definition}

Remarks that for the KL-Divergence, we shall write $\dkl{\pi(X)}{\eta(X)}$ and $\dklhor{\pi(X)}{\eta(X)}$ interchangeably. For both distance functions, we might drop their variables when the context is clear.

\begin{lemma}[Pinsker's Inequality] For any distributions $\pi$ and $\eta$, we have:
$$\|\pi - \eta\| = O\left(\sqrt{\dklhor{\pi}{\eta}}\right).$$
    
\end{lemma}

\section{Formalizing Communication Protocols}
\label{sec:formalize_protocols}

In this section, we define communication protocol and introduce various definitions which will be used throughout the paper. It is important to emphasize that we do not claim novelty regarding these definitions, propositions, theorems, or their proofs. However, they serve as the foundations for our work.

\subsection{Distributional View of Communication Protocols} 
\label{subsec:dist_view_protocols}

Recall that in the two-player communication model, the players' task is to compute a function $f: \mathcal{X} \times \mathcal{Y} \rightarrow \{0,1\}$ using a (randomized) protocol $\pi$ that dictates a sequence of messages $M$. Suppose that the input pair $(X,Y)$ is drawn from some distribution $\mu$. We can examine $\pi$ as a distribution over a set of random variables $(X,Y,\M)$, where $\pi(X,Y)$ represents the input distribution $\mu$, and $\M$ governs the public randomness $M^0$, and the sequence of messages $M^+ = (M^1, M^2, \ldots, M^r)$. More precisely, we interpret the \emph{standard} communication protocol as a distribution over these variables.

\begin{definition}[Standard Communication Protocols] Let $f: \mathcal{X} \times \mathcal{Y} \rightarrow \{0,1\}$, $\mu$ be an input distribution to $f$, and $\rho \in (0,1)$ be the error parameter. We say $\pi = (X,Y,\M)$ is a \emph{standard protocol} for computing $f$ on input distribution $\mu$ with probability $\rho$ iff
\begin{itemize}
    \item $\pi(x,y) = \mu(x, y)$
    \item $\M = (M^0, M^1, M^2,...,M^r)$ consists of public randomness $M^0$ and a sequence of messages $M^+ = (M^1, M^2,...,M^r)$ such that each $M^i$ only depends on $M^{<i}$ and the \underline{sender's} input. 
    \item Towards the end of the protocol, Alice and Bob output $M^r$ which correctly computes $f(x, y)$ with probability $\rho$.
\end{itemize}
\end{definition}
To expand on the last bullet, we assume that the last message sent by one of the players is their answer to $f(x, y)$. Any protocol can be converted to this form by having the first player who knows the answer send it to the other, adding only one additional bit to the message. We also assume that Alice sends the odd messages $(M^1,M^3,...)$ and Bob sends the even messages $(M^2, M^4,...)$.

A key characteristic of standard protocols is the restriction that each player generates a message based solely on their own inputs and the previously exchanged messages. Specifically, assuming Alice always speaks first, a standard protocol $\pi$ must satisfy the conditions $M^i \perp Y \mid X, M^{<i}$ for all odd $i$, and $M^i \perp X \mid Y, M^{<i}$ for all even $i$. In this work, we will explore an extended notion of standard protocols in which each message can be correlated with the receiver's inputs, referred to as a \emph{generalized protocol}.

\begin{definition}[Generalized Communication Protocols] Let $f: \mathcal{X} \times \mathcal{Y} \rightarrow \{0,1\}$, $\mu$ be an input distribution to $f$, and $\rho \in (0,1)$ be the error parameter. We say $\pi = (X,Y,\M)$ is a \emph{generalized protocol} for computing $f$ on input distribution $\mu$ with probability $\rho$ iff
\begin{itemize}
    \item $\pi(x,y) = \mu(x, y)$
    \item $\M = (M^0, M^1, M^2,...,M^r)$ consists of public randomness $M^0$ and a sequence of messages $M^+ = (M^1, M^2,...,M^r)$ such that each $M^i$ only depends on $M^{<i}$ and the \underline{both players'} input. 
    \item Towards the end of the protocol, Alice and Bob output $M^r$ which correctly computes $f(x, y)$ with probability $\rho$.
\end{itemize}
\end{definition}

\subsection{Information Costs and Information Complexity}
\label{subsec:info_costs}

We now define our resource of interest: \emph{information}.

\begin{definition}[Information Cost of a Protocol] For a (standard or generalized) protocol 
 $\pi = (X,Y,\M)$, we denote the \emph{(internal) information cost} of $\pi$ to be: 
 \begin{align*} \IC(\pi)  & = I(\M:X \mid YM_0) + I(\M:Y \mid XM_0) \\ & = I(M^+:X \mid YM_0) + I(M^+:Y \mid XM_0).
 \end{align*}
\end{definition}

To reason about the information cost, let us examine the first term. In Bob's view, at the end of the protocol he learns the message $M^+$, while already knows his own input $Y$ and the public randomness $M_0$. Hence, the amount of information that he gains of Alice's input $X$ is exactly is $I(M:X \mid YM_0)$. We also have the symmetric term for Alice's gain. In other words, the \emph{information cost} captures the amount of information that both parties learns from executing a protocol, or equivalently the amount of information that the protocol reveals to the players.

The following set of notions are borrowed from \cite{Braverman15}.

\begin{definition}[Distributional Information Complexity] Let $f$ be a $\{0,1\}$-valued function and $\varepsilon > 0$. Let $\mu$ be an input distribution. Then, the \emph{distributional information complexity} of $\pi$ of a function $f$ with error $\varepsilon$ and distribution $\mu$ is

$$\IC_{\mu}(f, \varepsilon) = \inf_{\pi; \Pr_{(x,y) \sim \mu} (\pi(x,y) \ne f(x,y) ) \leq \varepsilon} \IC(\pi).$$
\label{def:dist_ic}
\end{definition}

\begin{definition}[Max-Distributional Information Complexity]
    Let $f$ be a $\{0,1\}$-valued function and $\varepsilon > 0$. The \emph{max-distributional information complexity} of a function $f$ with error $\varepsilon$ is
    $$\IC_{\mathsf{D}}(f, \varepsilon) = \max_{\mu} \IC_{\mu}(f, \varepsilon).$$
\end{definition}

\begin{definition}[Information Complexity]
    Let $f$ be a $\{0,1\}$-valued function and $\varepsilon > 0$. The \emph{information complexity} of a function $f$ with error $\varepsilon$ is
    $$\IC(f, \varepsilon) = \inf_{\pi \text{ that errs w.p. at most $\varepsilon$ on any inputs}} \max_{\mu} \IC_{\mu}(\pi).$$
\label{def:ic}
\end{definition}

In other words, the \emph{information complexity} of $f$ is defined as the information cost of the best protocol that solves $f$ with a probability of failure at most $\varepsilon$, evaluated against its worst-case input distribution. By definition, it is trivial to see that $\IC(f, \varepsilon) \geq \IC_{\mathsf{D}}(f,\varepsilon)$. The following (minimax) theorem, by setting $\alpha = \frac{1}{2}$, implies that they are asymptotically-equivalent.

\begin{theorem}[Theorem 3.5 of \cite{Braverman15}]
    Let $f$ be any function, and let $\varepsilon > 0$. For each value of $\alpha \in (0,1)$, we have
    $$\IC(f, \frac{\varepsilon}{\alpha}) \leq \frac{\IC_{\mathsf{D}}(f, \varepsilon)}{1-\alpha}.$$
\label{thm:minimax_ic}
\end{theorem}
\subsection{Operations on Protocols}
\label{subsec:ops}

Here, we define a set of operations that can be applied on a protocol, turning it into another protocol(s) with some desirable properties. The first operation decomposes a protocol $\pi$ into two $\pi_0$ and $\pi_1$, each running on a smaller set of inputs.

\begin{restatable}[Binary Protocol Decomposition]{defi}{decomp}
Let $\pi = (X, Y, \M)$ be a generalized protocol where $\M$ consists of public randomness $M^0$ and a sequence of messages $M^+ = (M^1,M^2,...,M^r)$. Let $X = (X_0,X_1)$ and $Y = (Y_0, Y_1)$ be the partition of input coordinates. The \emph{binary protocol decomposition} of $\pi$ yields two protocols $\pi_0$ and $\pi_1$ with the following distributions:
\begin{itemize}
    \item $\pi_0 = (X_0, Y_0, \M^{\pi_0})$ where $ \M^{\pi_0} = (M^0, Y_1 \circ M^1, M^2,..., M^r)$
    \item $\pi_1 = (X_1, Y_1, \M^{\pi_1})$ where $ \M^{\pi_1} = (M^0 \circ X_0, M^1 , M^2,..., M^r)$
\end{itemize}  
 \label{def:protocol_decomposition}
\end{restatable}

It is worth noting that we will eventually consider a slight variant of decomposition which, roughly speaking, applies the decomposition to a \emph{conditional} distribution $\pi \mid W$ for some event $W$. 

The next operation captures sending one additional message on top of a protocol.

\begin{definition}[Appending Messages]
Let $\pi = (X,Y,\M)$ be a generalized protocol, and let $B$ be an arbitrary distribution representing an additional message. Denote $\pi \odot B$ to be a generalized protocol where the players execute $\pi$, followed by sending $B$. Distribution-wise, this can be written as $$\pi \odot B = (X,Y, (\M,B)).$$
\end{definition}

It is worth remarking that the distributions of $\pi$ and $\pi \odot B$ over the variables $(X,Y, \M, B)$ are in fact identical. Their only distinction is that $\pi \odot B$ includes $B$ as a part of the messages, while $\pi$ does not. Hence, we can also write the distribution of $\pi \odot B$ as:
$$(\pi \odot B)(X,Y,\M,B) := \pi(X,Y,\M) \cdot \pi(B \mid X,Y,\M).$$

Last but not least, given a generalized protocol, we can convert it into a new standard protocol via the following procedure.

\begin{restatable}[Standarization]{defi}{standardize}
Let $\pi = (X,Y,\M)$ be a generalized protocol, and let $\mu$ be an arbitrary input distribution. 
Say $\pi' = \stdz(\pi, \mu)$ is the \emph{standardization of $\pi$ with respect to $\mu$} iff $\pi'$ admits the following distribution:
$$\pi'(X,Y,\M) = \pi(M^0) \cdot \mu(X,Y) \cdot \prod_{\text{odd }i \geq 1}  \pi(M^i \mid X M^{<i}) \cdot \prod_{\text{even }i \geq 2}\pi(M^i \mid Y M^{<i}).$$    
\end{restatable}

A key observation is that $\pi' = \stdz(\pi, \mu)$ is a standard protocol. This follows from the definition: in each odd round $i$, Alice generates her message $M^i$ from the distribution $\pi(M^i \mid X M^{<i})$, which depends solely on her input and the previous messages. In other words, each of Alice's message in $\pi'$ disregards any correlation with Bob's inputs. A similar argument applies to Bob's messages.

\subsection{Rectangle Properties}
\label{subsec:rectangle}

The \emph{rectangle property} is a fundamental aspect of communication protocols. The following set of notions were introduced in \cite{Yu22}.

\begin{definition}[Rectangle Property]
Let $\pi = (X, Y, \M)$ be a generalized protocol. We say $\pi$ has the \emph{rectangle property} with respect to $\mu$ iff there exists nonnegative functions $g_1$ and $g_2$ such that
$$\pi(X, Y, \M) = \mu(X, Y) \cdot g_1(X, \M) \cdot g_2(Y, \M).$$
\end{definition}

\begin{fact} Standard communication protocols admit the rectangle properties.
\end{fact}

Recall that by Definition \ref{def:protocol_decomposition}, the protocol decomposition procedure splits inputs into two parts. To facilitate such partition, \cite{Yu22} also proposed the \emph{partial rectangle property}.
\begin{definition}[Partial Rectangle Property]
Let $\pi = (X, Y, \M)$ be a generalized protocol such that $X = (X_0, X_1)$ and $Y = (Y_0, Y_1)$. We say $\pi$ has the \emph{partial rectangle property} with respect to $\mu$ iff there exists nonnegative functions $g_1, g_2, g_3$ such that 
$$\pi(X,Y,\M) = \mu(X,Y) \cdot g_1(X, \M) \cdot g_2(Y,\M) \cdot g_3(X_0, Y_1, \M).$$
\label{def:partial_rec}
\end{definition}

It turns out that the partial rectangle property suffices to ensure the independence between inputs after the decomposition.

\begin{proposition}
    If a generalized protocol $\pi = (X,Y,\M)$ has the partial rectangle property with respect to $\mu = (\mu_0, \mu_1)$ where $\mu_0 \perp \mu_1$, then $X_1 \perp Y_0 \mid X_0Y_1\M$ in the distribution $\pi$. 
\label{prop:partial_rec_independent}
\end{proposition}

\begin{proof} It follows from the partial rectangular property that 
\begin{align*}
    \pi(X_1 Y_0 \mid X_0Y_1\M) & = \frac{\pi(X, Y, \M )}{\pi(X_0,Y_1,\M)} \\
    & = \frac{\mu(X,Y) \cdot g_1(X, M) \cdot g_2(Y,\M) \cdot g_3(X_0, Y_1, \M)}{\pi(X_0,Y_1,\M)} \\
    & = \left(\mu_0(X_0, Y_0) \cdot g_2(Y, \M)\right) \cdot \left(\mu_1(X_1, Y_1) \cdot g_1(X, \M) \cdot \frac{g_3(X_0, Y_1, \M)}{\pi(X_0,Y_1,\M)}\right).
\end{align*}

Now, conditioned on $X_0Y_1\M$, the first term only depends on $Y_0$ and the second term only depends on $X_1$. This proves the conditional independence.
\end{proof}

More importantly, the partial rectangle property of $\pi$ implies the rectangle property of its decomposed protocols if $\mu$ is a product distribution. The following lemma will be used recursively throughout our proofs.

\begin{lemma}
    Suppose that a generalized protocol $\pi = (X, Y, \M)$ has the partial rectangle property with respect to $\mu = (\mu_0, \mu_1)$ where $\mu_0 \perp \mu_1$. Let $\pi_0 = (X_0, Y_0, \M^{(\pi_0)})$ and $\pi_1= (X_1, Y_1, \M^{(\pi_1)})$ be generalized protocols obtained via decomposing $\pi$ (recall Definition \ref{def:protocol_decomposition}.) Then, $\pi_0$ has a rectangle property with respect to $\mu_0$, and $\pi_1$ has a rectangle property with respect to $\mu_1$.
\label{lem:decompose_then_rec}
\end{lemma}

\begin{proof} We forst prove the rectangle propertiy of $\pi_0$. Consider
\begin{align*}
    \pi_0(X_0, Y_0, \M^{(\pi_0)}) & = \pi(X_0, Y, \M) \\
    & = \sum_{X_1} \pi(X = (X_0, X_1), Y, \M) \\
    & = \sum_{X_1} \mu(X, Y) \cdot g_1(X, \M) \cdot g_2(Y,\M) \cdot g_3(X_0, Y_1, \M) \\
    & = \mu_0(X_0, Y_0) \cdot g_2(Y,\M) \cdot \left(\sum_{X_1} \mu_1(X_1, Y_1) \cdot g_1(X, \M) \cdot g_3(X_0, Y_1, \M)\right).
\end{align*}
Notice that the second term is a function of $Y_0, \M^{(\pi_0)}$. The third term is a function of $X_0, \M^{(\pi_0)}$. This concludes rectangle property of $\pi_0$. The proof of the rectangle property of $\pi_1$ also follows closely.
\end{proof}

\subsection{\texorpdfstring{$\theta$}{O}-cost and \texorpdfstring{$\gamma$}{r}-cost}
\label{subsec:costs}

The $\theta$-cost of a generalized protocol $\pi$ (with respect to $\mu$) roughly measure a combination of two distances: the closeness of $\pi$ to $\stdz(\pi, \mu)$, and the closeness of its input distribution of $\pi(X,Y)$ to $\mu$.\footnote{Our version of $\theta$-cost and $\gamma$-cost are equivalent to the logarithmic version of the $\theta$-cost and $\chi^2$-cost from \cite{Yu22}, respectively.}

\begin{definition}[Pointwise-$\theta$-cost]
For a generalized protocol $\pi = (X,Y,\M)$, input distribution $\mu$, and points $(X,Y,\M)$, denote the \emph{pointwise-$\theta$-cost} of $\pi$ with respect to $\mu$ at $(X,Y,\M)$ by the following quantity:
    \begin{align*}
    \theta_{\mu}(\pi @ X,Y,\M) & = \log \frac{\pi(X,Y \mid M^0)}{\mu(X,Y)} + \sum_{\text{odd } i \geq 1} \log{\frac{\pi(M^i \mid X, Y, M^{<i})}{\pi( M^i \mid X, M^{<i})}} + \sum_{\text{even } i \geq 2} \log{\frac{\pi(M^i \mid X, Y, M^{<i})}{\pi( M^i \mid Y, M^{<i})}} \\
    & = \log \left(\frac{\pi(X, Y, \M)}{ \pi(M^0) \cdot \mu(X,Y) \cdot \prod_{\text{odd } i \geq 1} \pi( M^i \mid X, M^{<i}) \cdot \prod_{\text{even } i \geq 2} \pi( M^i \mid Y, M^{<i})}\right) \\
    & = \log \frac{\pi(X,Y,\M)}{\eta(X,Y,\M)}
\end{align*}
where $\eta = \stdz(\pi, \mu)$ is the standardization of $\pi$ with respect to $\mu$.
\end{definition}

We can also define the $\theta$-cost of a protocol as follows.

\begin{definition}[$\theta$-cost]
Let $\pi = (X,Y,\M)$ be a generalized protocol. The \emph{$\theta$-cost} of $\pi$ with respect to $\mu$ is defined as
$$\theta_\mu(\pi) = \E_{(X,Y,\M) \sim \pi} \theta_{\mu}(\pi@ X,Y,\M).$$
\label{def:theta}
\end{definition}

By expanding the pointwise-$\theta$-cost, the following fact is immediate.

\begin{fact} Let $\pi = (X,Y,\M)$ be a generalized protocol, and let $\mu$ be an input distribution. Let $\eta = \stdz(\pi, \mu)$ be the standardization of $\pi$ with respect to $\mu$. Then, we have
$$\theta_\mu(\pi) = \dkl{\pi(X,Y,\M)}{\eta(X,Y,\M)}.$$
\label{fact:theta_dkl}
\end{fact}

\begin{observation} If $\pi$ is a standard protocol with input distribution $\pi(x,y) = \mu$, then $\theta_{\mu}(\pi) = 0$.
\label{obs:std_theta_zero}
\end{observation}
It turns out that the pointwise-$\theta$-cost admit linearity when undergoing a decomposition.

\begin{restatable}[Linearity of pointwise-$\theta$-cost]{lem}{ptwisetheta} Let $\pi = (X,Y,\M)$ be a generalized protocol with the partial rectangle property with respect to $\mu = (\mu_0, \mu_1)$. Let $\pi_0 = (X_0, Y_0, \M^{(\pi_0)})$ and $\pi_1= (X_1, Y_1, \M^{(\pi_1)})$ be generalized protocols obtained via decomposing $\pi$ (recall Definition \ref{def:protocol_decomposition}.)  Then, for any point $(X,Y, \M)$, we have
$$\theta_{\mu}( \pi @ X,Y,\M) = \theta_{\mu_{0}}( \pi_{0} @ X_0,Y_0,\M^{(\pi_0)}) + \theta_{\mu_{1}}( \pi_{1} @ X_1,Y_1,\M^{(\pi_1)}).$$
\label{lem:pointwise_linear_theta}   
\end{restatable}

Next, we discuss the $\gamma$-cost which roughly measures the amount of information that each parties reveals about their inputs via the messages.

\begin{definition}[Pointwise-$\gamma$-cost] For a generalized protocol $\pi = (X,Y,\M)$, input distribution $\mu$, and points $(X,Y,\M)$, denote the \emph{pointwise-$\gamma$-cost} of $\pi$ with respect to $\mu$ at $(X,Y,\M)$ by the following quantities:
\begin{align*}
    \gamma_{\mu, A}(\pi @ X,Y,\M) & = \log \frac{\pi(X \mid Y \M)}{\mu(X \mid Y)}\\
    \gamma_{\mu, B}(\pi @ X,Y,\M) & = \log \frac{\pi(Y \mid X \M)}{\mu(Y \mid X)}.
\end{align*}
\end{definition}

Similar to the $\theta$-cost, we can define the $\gamma$-cost of a protocol. Additionally, the pointwise-$\gamma$-cost also admit linearity when undergoing a decomposition.

\begin{definition}[$\gamma$-cost]
Let $\pi$ be a generalized protocol. The \emph{$\gamma$-cost} of $\pi$ with respect to $\mu$ is defined as
$$\gamma_\mu(\pi) = \E_{(X,Y,\M) \sim \pi} \gamma_{\mu, A}(\pi @ X,Y, \M) + \E_{(X,Y,\M) \sim \pi} \gamma_{\mu, B}(\pi @ X,Y, \M).$$
\label{def:delta}
\end{definition}

\begin{restatable}[Linearity of pointwise-$\gamma$-cost]{lem}{ptwisedelta} Let $\pi = (X,Y,\M)$ be a generalized protocol with the partial rectangle property with respect to $\mu = (\mu_0, \mu_1)$. Let $\pi_0 = (X_0, Y_0, \M^{(\pi_0)})$ and $\pi_1= (X_1, Y_1, \M^{(\pi_1)})$ be generalized protocols obtained via decomposing $\pi$ (recall Definition \ref{def:protocol_decomposition}.)  Then, for any point $(X,Y, \M)$, we have
\begin{align*}
    \gamma_{\mu, A}(\pi @ X,Y, \M) & =  \gamma_{\mu_0, A}(\pi_0 @ X_0,Y_0, \M^{(\pi_0)}) + \gamma_{\mu_1, A}(\pi_1 @ X_1,Y_1, \M^{(\pi_1)}) \\
    \gamma_{\mu, B}(\pi @ X,Y, \M) & =  \gamma_{\mu_0, B}(\pi_0 @ X_0,Y_0, \M^{(\pi_0)}) + \gamma_{\mu_1, B}(\pi_1 @ X_1,Y_1, \M^{(\pi_1)}).
\end{align*}
\label{lem:pointwise_linear_delta}
\end{restatable}

The proofs of Lemma \ref{lem:pointwise_linear_theta} and \ref{lem:pointwise_linear_delta} are deferred to the appendix.

\section{Useful Facts and Lemmas}
\label{sec:lemmas}

In this section, we present key facts and lemmas which will be needed later in the paper.

\begin{fact} For any $p \in (0,1)$, denote $\cH(p) = p \log \frac{1}{p} + (1-p) \log{\frac{1}{1-p}}$. Then, we have:
\begin{enumerate}
    \item $\cH$ is concave.
    \item For any $p \in (0,\frac{1}{2}]$, we have $\cH(p) \leq 2p \log \frac{1}{p}$.
\end{enumerate}
\end{fact}

\begin{lemma} Let $A,B,C$ be variables and $E$ be an event which occurs with probability $p$. Then, 
    $$I(A:B \mid CE) \leq \frac{1}{p} \cdot \left[I(A:B \mid C) + \cH(p)\right].$$
\label{lem:ic_condition}
\end{lemma}

\begin{proof} Observe that $I(A:B \mid C\mathbbm{1}_E) \geq p \cdot I(A:B \mid CE)$. Moreover, we have
\begin{align*}
    I(A:B \mid C\mathbbm{1}_E) & = H(A \mid C \mathbbm{1}_E) - H(A \mid BC \mathbbm{1}_E) \\
    & \leq H(A \mid C) - H(A \mid BC\mathbbm{1}_E) \\
    & = H(A \mid C) - \left[ H(A \mathbbm{1}_E\mid BC) - H(\mathbbm{1}_E \mid BC) \right] \\
    & \leq H(A \mid C) - \left[ H(A \mid BC) - H(\mathbbm{1}_E)\right] \\
    & = I(A:B \mid C) + H(\mathbbm{1}_E) \\
    & = I(A:B \mid C) + \cH(p).
\end{align*}
Combining two inequalities conclude the proof.
\end{proof}

\begin{lemma} Let $\pi$ and $\eta$ be probability distributions and $E$ be an event in the same probability space. Then, we have
    $$\dkl{\pi( X \mid E)}{\eta(X)} \leq \frac{1}{\pi(E)} \cdot  \left[\dkl{\pi(X)}{\eta(X)} + \cH(\pi(E))\right].$$
\label{lem:expected_kl_helper}
\end{lemma}

\begin{proof}
Consider
    \begin{align*}
        \dkl{\pi(X)}{\eta(X)} & = \sum_{x} \pi(x) \cdot \log{\frac{\pi(x)}{\eta(x)}} \\ 
        & = \pi(E) \cdot \left(\sum_{x} \pi(x \mid E) \cdot \log{\frac{\pi(x)}{\eta(x)}} \right)+ \left(\sum_{x} \pi(x\overline{E}) \cdot \log{\frac{\pi(x)}{\eta(x)}}\right).
    \end{align*}

Note that 
    \begin{align*}
        \sum_{x} \pi(x \mid E) \log{\frac{\pi(x)}{\eta(x)}} & \geq \sum_{x} \pi(x \mid E)  \log{\frac{\pi(x E)}{\eta(x)}} \\
        & = \sum_{x} \pi(x \mid E) \log{\frac{\pi(x \mid E) \cdot \pi(E)}{\eta(x)}} \\
        & = \log \pi(E) \cdot \sum_x \pi(x\mid E) +  \sum_{x} \pi(x \mid E) \log{\frac{\pi(x \mid E)}{\eta(x)}} \\
        & = \log \pi(E) + \dkl{\pi(X\mid E)}{\eta(X)}
    \end{align*}

    Also via Jensen,
    \begin{align*}
        \sum_{x} \pi(x\overline{E}) \cdot \log{\frac{\pi(x)}{\eta(x)}} & \geq  \sum_{x} \pi(x\overline{E}) \cdot \log{\frac{\pi(x\overline{E})}{\eta(x)}} \\
        & \geq \pi(\overline{E}) \cdot \log{(\pi(\overline{E})}) \tag{Jensen's}
    \end{align*}

    Combining the two inequalities yield the result.
\end{proof}

The next lemma offers its extension.

\begin{lemma} Let $\pi$ and $\eta$ be probability distributions and $E$ be an event in the same probability space. Then, we have
    $$\E_{y \sim \pi \mid E}\dkl{\pi( X \mid yE)}{\eta(X)} \leq \frac{1}{\pi(E)} \cdot  \left[\E_{y \sim \pi }\dkl{\pi(X \mid y)}{\eta(X)} + \cH(\pi(E))\right].$$
\label{lem:expected_kl}
\end{lemma}

\begin{proof}
    For any value of $y$, we have
    \begin{align*}
        \pi(y \mid E) \cdot \dkl{\pi(X \mid yE)}{\eta(X)} & \leq \pi(y \mid E) \cdot \left[\frac{1}{\pi(E \mid y)} \cdot \dkl{\pi(X \mid y)}{\eta(X)} + \cH(\pi(E \mid y))\right] \tag{Lemma \ref{lem:expected_kl_helper}}\\
        & = \frac{\pi(y)}{\pi(E)} \cdot \dkl{\pi(X \mid y)}{\eta(X)} + \pi(y \mid E) \cdot \cH(\pi(E \mid y))
    \end{align*}
    Summing over $y$, we have
    \begin{align*}
        \E_{y \sim \pi \mid E}\dkl{\pi(X \mid yE)}{\eta(X)}  & \leq \frac{1}{\pi(E)} \cdot \E_{y \sim \pi}\dkl{\pi(X \mid y)}{\eta(X)} + \sum_{y} \pi(y \mid E) \cdot \cH(\pi(E \mid y)) \\
        & \leq \frac{1}{\pi(E)} \cdot \E_{y \sim \pi}\dkl{\pi(X \mid y)}{\eta(X)}  + \frac{1}{\pi(E)}\sum_{y} \pi(y) \cdot \cH(\pi(E \mid y)) \tag{$\pi(y \mid E) \leq \frac{\pi(y)}{\pi(E)}$} \\
        & \leq \frac{1}{\pi(E)} \cdot \E_{y \sim \pi}\dkl{\pi(X \mid y)}{\eta(X)}  + \frac{1}{\pi(E)} \cdot \cH\left(\sum_{y} \pi(y) \cdot \pi(E \mid y) \right)\tag{Jensen's} \\
        & = \frac{1}{\pi(E)} \cdot \E_{y \sim \pi}\dkl{\pi(X \mid y)}{\eta(X)}  + \frac{1}{\pi(E)} \cdot \cH\left(\pi(E)\right)
    \end{align*}
Rearranging it yields the result.
\end{proof}

The following lemma will be used for analyzing the losses incurred by protocol standardization. It appears as the ``correlated sampling'' in  Lemma 7.5 of \cite{Rao_Yehudayoff_2020}. We present our proof in the Appendix.

\begin{restatable}{lem}{couplinglemma}
Let $\mu$ and $\mu'$ be distributions over supports $\mathcal{A}$. There exists a random process such that at the end of the process, we obtain $a$ and $a'$ such that $a$ distributes according to $\mu$, $a'$ distributes according to $\mu'$, and the probability that $a \ne a'$ is at most $O(\|\mu - \mu'\|)$.
\label{lem:coupling_lemma}
\end{restatable}

\section{Obtaining a Generalized Protocol for \texorpdfstring{$f$}{f}}
\label{sec:decomposition}

For the next two sections, we will prove Lemma \ref{lem:main_lemma}. We assume the setup of the lemma as follows. Let $\pi = (X,Y,\M)$ where $\M = (M^0, M^1,...,M^r)$ be a standard protocol that computes $f^{\oplus n}$ over the input distribution $\mu^n$ with disadvantage $\varepsilon = 1/5$.\footnote{This disadvantage $\varepsilon = 1/5$ corresponds to the error probability of $1/10$ for $\pi$ over $\mu^n$.} Denote the input $(X,Y) \sim \mu^n$ by $X = (X_1,...,X_n)$ and $Y = (Y_1,...,Y_n)$ such that each $(X_i, Y_i)$ is independently drawn from $\mu$. Also denote the information cost of $\pi$ by $\mathcal{I}$.

We will eventually show that there exists a standard protocol $\eta$ that computes $f$ over $\mu$, errs with probability $1/poly(n)$, and has information cost approximately $O(\mathcal{I}/n+1)$. In this section, we address a relaxed version of this statement: showing that there exists a \emph{generalized} protocol that \emph{almost} achieves these properties while remaining \emph{close} to being standard.

\subsection{Refinements of Conditional Decomposition}
\label{subsec:one_step}

We begin by discussing a revision of the \emph{conditional decomposition} described in Section \ref{sec:tech_overview}, which we will eventually apply recursively. For simplicity, in this subsection, we demonstrate only the top level of the recursion.

\paragraph{Conventions.} We assume that $n = 2^m$ is a power of two, and label the $n$ coordinates with $\{0,1\}^m$. Additionally, for any string $S \in \{0,1\}^{\leq m}$, we denote $\mu_S$, $X_S$, and $Y_S$ as the input distributions, Alice's input coordinates, and Bob's input coordinates for the $2^{m - |S|}$ instances prefixed with $S$, respectively. Let $\alpha \in (0,1)$ be a constant such that $\alpha = 1-2\varepsilon$.\footnote{As long as $\pi$ errs on $\mu^n$ with probability at most $1/4-O(1)$, we have $\varepsilon$ is twice the error, which is at most $1/2-O(1)$ Then, we are able to set $\alpha = 1-2\varepsilon = \Omega(1)$.}. Let $\tau \in (0,1)$ be a constant such that $\frac{1+\sqrt{\alpha}}{2} = 2^{-\tau}$.

Recall the binary protocol decomposition procedures.

\decomp*

\remark{In protocol $\pi_0$, we can interpret the first two messages, $(M^0, Y_1 \circ M^1)$, as Alice drawing $Y_1$ and prepending it to her first message to Bob, allowing him to recover $Y_1$. A crucial observation is that, distribution-wise, the distribution of $M^{\pi_0}$ remains unchanged if the first two terms are replaced by $(M^0 \circ Y_1, M^1, \ldots, M^r)$. This perspective, in fact, corresponds to the protocol decomposition discussed earlier in the overview. These interpretations are equivalent in terms of \emph{information} cost; however, their \emph{communication} costs differ due to $Y_1$ becoming a part of the transcript.}

\paragraph{Conditional Decomposition of $\pi$.} We begin by discussing the first level of our decomposition, where we use the protocol $\pi$ to derive two \emph{generalized} protocols, $\pi_0$ and $\pi_1$, each of which computes the function $f^{\oplus n/2}$. We highlight that the process can be viewed as two steps: (1) apply the binary decomposition to $\pi \mid W$, for some event $W$, to obtain $\widetilde{\pi}_0$ and $\widetilde{\pi}_1$, and (2) complete the each protocol by appending the ``final bit".

\begin{figure}[H]
    \centering
    \begin{tabular}{|p{15.5cm}|}
    \hline ~\\
    \multicolumn{1}{|c|}{\textbf{Conditional Decomposition of $\pi = (X,Y,\M)$ into $\pi_0$ and $\pi_1$.}} \\ 
    \\ \hline
    \begin{enumerate}
        \item Let $\mathcal{G}$ be a set of $(X_0, Y_1, \M)$ such that $\adv^{\pi}(f^{\oplus n}(X,Y) \mid X_0, Y_1, \M) \geq \alpha$.
        \item Let $W := \left[(X_0, Y_1, M) \in \mathcal{G}\right]$ be an event that $(X_0,Y_1, \M)$ yields $\adv^{\pi}(f^{\oplus n}(X,Y) \mid X_0, Y_1, \M) \geq \alpha$.
        \item Apply the Binary Decomposition to $\pi \mid W$. This yields two protocols, each of which computes $f^{\oplus n/2}$: 
        \begin{itemize}
            \item $\widetilde{\pi}_0 = (X_0, Y_0, \widetilde{\M}_0)$ where $\widetilde{\M}^0_0 = M^0$ and $\widetilde{\M}^+_0 = (Y_1 \circ M^1, M^2,...,M^r)$
            \item $\widetilde{\pi}_1 = (X_1, Y_1, \widetilde{\M}_1)$ where $\widetilde{\M}^0_1 = M^0 \circ X_0$ and $\widetilde{\M}^+_1 = (M^1, M^2,...,M^r)$
        \end{itemize}
        \item[4a.] The protocol $\pi_0 = (X_0, Y_0, \M_0)$ is obtained as follows.
        \begin{enumerate}
            \item[$(i)$] Upon finishing $\widetilde{\pi}_0$, Alice determines the more-likely bit $B_0$ among the posterior distribution $$\widetilde{\pi}_0(f^{\oplus n/2}(X_0,Y_0) \mid X_0, \widetilde{\M}_0).$$ 
            \item[$(ii)$] Alice sends $B_0$ to Bob, and both players declare $B_0$ as the answer.
        \end{enumerate}
        Equivalently, distribution-wise we can write $\pi_0 := \widetilde{\pi}_0 \odot B_0$.
        
        \item[4b.] The protocol $\pi_1 = (X_1, Y_1, \M_1)$ is obtained as follows.
        \begin{enumerate}
            \item[$(i)$] Upon finishing $\widetilde{\pi}_1$, Bob determines the more-likely bit $B_1$ among the posterior distribution $$\widetilde{\pi}_1(f^{\oplus n/2}(X_1,Y_1) \mid Y_1, \widetilde{\M}_1).$$ 
            \item[$(ii)$] Bob sends $B_1$ to Alice, and both players declare $B_1$ as the answer.
        \end{enumerate}
        Equivalently, distribution-wise we can write $\pi_1 := \widetilde{\pi}_1 \odot B_1$.
    \end{enumerate}
    \\
    \hline
    \end{tabular}
    \caption{The ``conditional'' decomposition  of $\pi$ into $\pi_{0}$ and $\pi_1$.}
    \label{figure:decompose_protocols}
\end{figure}

As an important remark, we observe the distribution of $\pi_0(X_0, Y_0, \M_0)$ and $\widetilde{\pi}_0(X_0,Y_0, M_0, B_0)$ are in fact identical. However, their differences lies within the language of communication protocol: that $\pi_{0}$ contains the final message $B_0$ while $\widetilde{\pi}_0$ does not.

Notice further that in the procedure of $\widetilde{\pi}_0$, Alice knows $X_0$ and $\widetilde{\M}_0$. Therefore, she can determine $B_0$ by herself. In other words, we have $B_0 \perp Y_0 \mid X_0 \widetilde{\M}_0$ in the the distribution of $\widetilde{\pi}_0$. With the same reasoning, we have $B_1 \perp X_1 \mid Y_1\widetilde{\M}_1$ in the the distribution of $\widetilde{\pi}_1$.

Following our decomposition, it can be shown that $\pi_0$ and $\pi_1$ has the rectangle property with respect to $\mu_0$ and $\mu_1$ respectively.

\begin{claim} If $\pi$ has a rectangle property with respect to $\mu_{\emptyset}$, then $\pi_0$ has a rectangle property with respect to $\mu_0$, and $\pi_1$ has a rectangle property with respect to $\mu_1$.
\label{clm:partial_rec}
\end{claim}

\begin{proof} First we will show that  $\pi \mid W$ has a partial rectangle property with respect to $(\mu_0, \mu_1)$.  Due to the rectangle property of $\pi$, let the functions $g_1,g_2$ be such that at any point $(X,Y,\M)$, we have 
$$\pi(X,Y,\M) = \mu(X,Y) \cdot g_1(X,\M) \cdot g_2(Y,\M).$$

Then, we have
\begin{align*}
\pi(X,Y,\M \mid W) & = \pi(X,Y,\M) \cdot \frac{ \pi(W \mid X,Y,\M)}{\pi(W)} \\
& = \mu(X,Y) \cdot g_1(X,\M) \cdot g_2(Y,\M) \cdot \frac{\mathbbm{1}[(X_0, Y_1, \M) \in \mathcal{G}]}{\pi(W)}.
\end{align*}
Note that the last term can be interpreted as a function of $(X_0,Y_1, \M)$. Therefore, $\pi \mid W$ has the partial rectangle property. 

Lemma \ref{lem:decompose_then_rec} then implies that $\widetilde{\pi}_0$ has a rectangle property with respect to $\mu_0$. Let the functions $g_3,g_4$ be such that for any points $(X_0, Y_0, \widetilde{\M}_0)$, we have $$\widetilde{\pi}_0(X_0,Y_0,\widetilde{\M}_0) = \mu(X_0, Y_0) \cdot g_3(X_0, \widetilde{\M}_0) \cdot g_4(Y_0, \widetilde{\M}_0).$$ Then, we have
\begin{align*}
{\pi}_0(X_0,Y_0, \M_0) & = \widetilde{\pi}_0(X_0,Y_0, \widetilde{\M}_0, B_0) \\
& = \widetilde{\pi}_0(X_0,Y_0,\widetilde{\M}_0) \cdot \widetilde{\pi}_0(B_0 \mid X_0,Y_0,\widetilde{\M}_0) \\
& = \mu(X_0, Y_0) \cdot g_3(X_0, \widetilde{\M}_0) \cdot g_4(Y_0, \widetilde{\M}_0) \cdot \widetilde{\pi}_0(B_0 \mid X_0,Y_0,\widetilde{\M}_0) \\
&  = \mu(X_0, Y_0) \cdot \left(g_3(X_0, \widetilde{\M}_0) \cdot \widetilde{\pi}_0(B_0 \mid X_0,\widetilde{\M}_0) \right)\cdot g_4(Y_0, \widetilde{\M}_0) \tag{$B_0 \perp Y_0 \mid X_0\M_0$ in $\widetilde{\pi}_0$}
\end{align*}
Recall that $\M_0$ consists of $\widetilde{\M}_0$ and $B_0$. Therefore, the second term is a function of $(X_0, M_0)$, and the third term is a function of $(Y_0, M_0)$. This proves that $\pi_0$ has a rectangle property with respect to $\mu_0$. The same proof applies for $\pi_1$.
\end{proof}

Now, we discuss the decomposition of advantages when a protocol undergoes the conditional decomposition. As a starting point, \cite{Yu22} observed the following.

\begin{claim}[Multiplicity of Advantages.] Suppose that $\pi$ has a rectangle property with respect to $\mu_\emptyset$. For any value of $(X_0,Y_1,\M)$, we have $$\adv^{\pi}(f^{\oplus n/2}(X_0, Y_0) \mid X_0, Y_1, \M) \cdot \adv^{\pi}(f^{\oplus n/2}(X_1, Y_1) \mid X_0, Y_1, \M) =  \adv^{\pi}(f^{\oplus n}(X,Y) \mid X_0, Y_1, \M).$$
\label{claim:mult_adv}
\end{claim}
\begin{proof} The rectangle property of $\pi$ implies $X_1 \perp Y_0 \mid X_0Y_1\M$. Thus, we have $f^{\oplus n/2}(X_0, Y_0) \perp f^{\oplus n/2}(X_1, Y_1) \mid X_0 Y_1 \M$ because first term only depends on $Y_0$ and the second term only depends on $X_1$. By Fact \ref{fact:xor_ind_bits}, the proof is concluded.
\end{proof}

For brevity, we denote random variables over the randomness of $(X_0, Y_1, \M)$:
\begin{align*}
    Z  & := \adv^{\pi}\left(f^{\oplus n}(X,Y) \mid X_0, Y_1, \M\right) \\
    A_0 & := \adv^{\pi}\left(f^{\oplus n/2}(X_0, Y_0) \mid X_0, Y_1, \M\right) \\
    A_1 & := \adv^{\pi}\left(f^{\oplus n/2}(X_1, Y_1) \mid X_0, Y_1, \M\right)
\end{align*}

With these notions, we realize that $W$ is, in fact, simply the event that $(Z \geq \alpha)$. Furthermore, Claim \ref{claim:mult_adv} can be expressed as $A_0 A_1 = Z$. According to the decomposition outlined in Figure \ref{figure:decompose_protocols}, the advantage of $\pi_0$ is $\mathbb{E}(A_0 \mid W)$. To see this, let us try to understand the advantage of $\pi_0$ from Alice's side. Let $(X_0, Y_1, \M)$ be arbitrary point which occurs with probability $\pi(X_0, Y_1, \M \mid W)$. Since Alice knows $(X_0, Y_1, \M)$ and outputs the more-likely answer $B_0$, she can compute $f^{\oplus n/2}(X_0, Y_0)$ correctly with probability 
$$\frac{1}{2} + \frac{\adv^{\widetilde{\pi}_0}(f^{\oplus n/2}(X_0,Y_0) \mid X_0, \widetilde{\M}_0)}{2} \ = \ \frac{1}{2} + \frac{\adv^{\pi}(f^{\oplus n/2}(X_0, Y_0) \mid X_0, Y_1, \M, W)}{2}.$$
Therefore, the advantage of $\pi_0$ from Alice's perspective is
$$\sum_{(X_0, Y_1, \M)} \pi(X_0, Y_1, \M \mid W) \cdot \adv^{\pi}(f^{\oplus n/2}(X_0, Y_0) \mid X_0, Y_1, \M, W) = \mathbb{E}(A_0 \mid W).
$$
Finally, by sending Bob an additional bit indicating her answer to $f^{\oplus n/2}(X_0, Y_0)$, Bob can achieve the same advantage. A similar argument applies to $\pi_1$, where its advantage is $\mathbb{E}(A_1 \mid W)$.

Denote $\varepsilon_0 = 1-\adv^{\pi_0}(f^{\oplus n/2}(X_0, Y_0)) = 1 - \E(A_0 \mid W)$ and $\varepsilon_1 = 1-\adv^{\pi_1}(f^{\oplus n/2}(X_1, Y_1)) = 1 - \E(A_1 \mid W)$ to be the disadvantages of $\pi_0$ and $\pi_1$ respectively. Again, by letting Alice output the more-likely, the protocol $\pi_0$ errs with probability $$\frac{1}{2} - \frac{\adv^{\pi_0}(f^{\oplus n/2}(X_0, Y_0))}{2} = \frac{\varepsilon_0}{2}$$
and vice-versa for Bob in $\pi_1$. Thus, it is intuitive to use $\varepsilon_0$ and $\varepsilon_1$ as a proxy to the error of $\pi_0$ and $\pi_1$ respectively.

The following claim shows that the sum of $\varepsilon_0$ and $\varepsilon_1$ cannot be too large.

\begin{claim}
    $\varepsilon_0 + \varepsilon_1 \leq \frac{2}{1+\sqrt{\alpha}} \cdot \left(1 - \E(Z \mid W)\right).$
\label{clm:disadv_bounds}
\end{claim}

\begin{proof} Conditioned on $W$ (meaning $Z \geq \alpha$) we derive:
    \begin{equation*}
        A_0 + A_1 \geq 2\sqrt{A_0A_1} = 2\sqrt{Z} \geq 2 \cdot \frac{Z + \sqrt{\alpha}}{1+\sqrt{\alpha}}
    \label{eq:decompose_disadvantage}
    \end{equation*}
where the last inequality is equivalent to $(\sqrt{Z}-1)(\sqrt{Z}-\sqrt{\alpha}) \leq 0$. Taking an expectation conditioned on $W$, the inequality becomes:
$$\E(A_0 \mid W) + \E(A_1 \mid W) \geq 2 \cdot \frac{\E(Z \mid W) + \sqrt{\alpha}}{1+\sqrt{\alpha}}.$$ Recall that $\varepsilon_0 = 1 - \E(A_0 \mid W)$ and $\varepsilon_1 = 1 - \E(A_1 \mid W)$. Rearranging it concludes the proof.
\end{proof}

The following claim shows that the conditioning event $W = (Z \geq \alpha)$ occurs with substantial probability.

\begin{claim}
    $\pi(W) \geq \frac{1-\varepsilon - \alpha}{\E(Z \mid W) - \alpha}.$
\label{clm:prob_bounds} 
\end{claim}

\begin{proof} By definition, $\overline{W}$ implies $Z \leq \alpha$. We also have $\E(Z) \geq \adv^{\pi}(f^{\oplus n}(X,Y)) = 1- \varepsilon$ via the triangle inequality. Then, we have
\begin{align*}
    1-\varepsilon \leq \E(Z) & = \pi(W) \cdot \E(Z \mid  W) + \pi(\overline{W})\cdot \E(Z \mid \overline{W}) \\
    & \leq \pi(W ) \cdot \E(Z \mid W) + \left(1-\pi(W)\right) \cdot \alpha
\end{align*}
Rearranging it concludes the proof.
\end{proof}

\subsection{Obtaining Generalized Protocol(s) for \texorpdfstring{$f$}{f}} 

We now describe the \emph{conditional decomposition} that is applied to $\pi_S$ for an arbitrary $S \in \{0,1\}^{\leq m-1}$ (Figure \ref{figure:protocols_piS0_piS1}), which is simply a succinct generalization of Figure \ref{figure:decompose_protocols}. Precisely, the conditional decomposition of $\pi_S$ yields two protocols, $\pi_{S0}$ and $\pi_{S1}$, each of which computes $f^{\oplus}$ on half of the inputs that $\pi_S$ handles. Specifically, for an event $W_S$ (to be defined shortly), we decompose $\pi_S \mid W_S$ according to Definition \ref{def:protocol_decomposition}, yielding two protocols, $\widetilde{\pi}_{S0}$ and $\widetilde{\pi}_{S1}$. To derive the protocol $\pi_{S0}$, we further let Alice compute the more likely answer of $f^{\oplus}(X_{S0}, Y_{S0})$ based on her knowledge, denoted $B_{S0}$, and send it to Bob. This bit $B_{S0}$ serves as the players' answer to $f^{\oplus}(X_{S0}, Y_{S0})$ in $\pi_{S0}$. Equivalently, we can interpret the protocol $\pi_{S0}$ as appending $B_{S0}$ to $\widetilde{\pi}_{S0}$, i.e., $\pi_{S0} := \widetilde{\pi}_{S0} \odot B_{S0}$. The protocol $\pi_{S1}$ is obtained analogously.

\begin{figure}[H]
    \centering
    \begin{tabular}{|p{15cm}|}
    \hline ~\\
    \multicolumn{1}{|c|}{\textbf{Conditional Decomposition of $\pi_S$ into $\pi_{S0}$ and $\pi_{S1}$.}} \\
    \\ \hline
    \begin{enumerate}
        \item Let $\widetilde{\pi}_{S0} = (X_{S0}, Y_{S0}, \widetilde{\M}_{S0})$ and $\widetilde{\pi}_{S1} = (X_{S1}, Y_{S1}, \widetilde{\M}_{S1})$ be the protocols obtained from decomposing $\pi_S \mid W_S$ via the procedure given in Definition \ref{def:protocol_decomposition}.
        \item Let $B_{S0}$ be the more-likely value of $f^{\oplus}(X_{S0}, Y_{S0})$ in the distribution $\widetilde{\pi}_{S0}(f^{\oplus}(X_{S0},Y_{S0}) \mid X_{S0} \widetilde{\M}_{S0})$ as computed by Alice. The protocol $\pi_{S0}$ is defined to be $\pi_{S0} := \widetilde{\pi}_{S0} \odot B_{S0}$.
    \item  Let $B_{S1}$ be the more-likely value of $f^{\oplus}(X_{S1}, Y_{S1})$ in the distribution $\widetilde{\pi}_{S1}(f^{\oplus}(X_{S1},Y_{S1}) \mid Y_{S1} \widetilde{\M}_{S1})$ as computed by Bob. The protocol $\pi_{S1}$ is defined to be $\pi_{S1} := \widetilde{\pi}_{S1} \odot B_{S1}$.
    
    \end{enumerate}
    \\
    \hline
    \end{tabular}
    \caption{The conditional decomposition procedure of $\pi_S$ into $\pi_{S0}$ and $\pi_{S1}$}
    \label{figure:protocols_piS0_piS1}
\end{figure}

We remark that the event $W_S$ is to be defined shortly via Table \ref{table:params_S}. The following observation is immediate, as by the end of $\widetilde{\pi}_{S0}$, Alice knows $X_{S0}$ and $\widetilde{M}_{S0}$, and vice versa for Bob.
\begin{observation} For any $S$, we have $B_{S0} \perp Y_{S0} \mid X_{S0} \widetilde{\M}_{S0}$ in the distribution of $\widetilde{\pi}_{S0}$, and $B_{S1} \perp X_{S1} \mid Y_{S1} \widetilde{\M}_{S1}$ in the distribution of $\widetilde{\pi}_{S1}$.
\label{obs:bs_ind}
\end{observation}

Next, we describe our \emph{recursive procedure} (Figure \ref{figure:recursive_decomposition}) which applies the conditional decomposition in lexicographical order: for $k = 0, 1, \ldots, m-1$, we split $2^k$ protocols $\{\pi_S\}_{|S| = k}$ into $2^{k+1}$ protocols $\{\pi_{S'}\}_{|S'| = k+1}$. We begin with $k = 0$, which corresponds to the single protocol $\pi_\emptyset = \pi$ for $f^{\oplus n}$. By the end of this process, after completing round $k = m-1$, we will have $n$ protocols $\{\pi_S\}_{|S| = m}$ for $f$. We collectively refer to the protocols $\{\pi_S\}_{|S| = k}$ as the \emph{level $k$}.

\begin{figure}[H]
    \centering
    \begin{tabular}{|p{15cm}|}
    \hline ~\\
    \multicolumn{1}{|c|}{\textbf{Recursive Procedure $\mathcal{P}$}} \\
    \\ \hline
    \begin{enumerate}
        \item set $\pi_{\emptyset}$ to $\pi$
        \item \textbf{for} each $k = 0,1,...,m-1$,
        \item \hspace{5mm} \textbf{for} each $S \in \{0,1\}^k$, 
        \item \hspace{10mm} apply the conditional decomposition from Figure \ref{figure:protocols_piS0_piS1} on $\pi_S$ to obtain $\pi_{S0}$ and $\pi_{S1}$
    \end{enumerate}
    \\
    \hline
    \end{tabular}
    \caption{A recursive procedure begins with a standard protocol $\pi$ for computing $f^{\oplus n}$ and ultimately yields $n$ generalized protocols $\{\pi_S\}_{|S| = m}$, each of which computes $f$.}
    \label{figure:recursive_decomposition}
\end{figure}

For any $S$, we define the following set of parameters related to the protocol $\pi_S$. Note that the case of $S = \emptyset$ corresponds to the first level of conditional decomposition discussed in the previous subsection. We also note that when $S = \emptyset$, we might drop the subscript $S$. By doing that, the notations become consistent with our earlier discussions (e.g. $\pi_\emptyset$ becomes $\pi$, $X_\emptyset$ becomes $X$, or $\varepsilon_{\emptyset}$ becomes $\varepsilon$, etc.)

\begin{table}[H]
    \centering
    \begin{tabular}{|p{15cm}|}
    \hline ~\\
    \multicolumn{1}{|c|}{\textbf{Parameters related to a generalized protocol $\pi_S = (X_S, Y_S, \M_S)$}} \\
    \\ \hline
    \begin{itemize}
        \item $(X_S,Y_S)$ denotes the input of $2^{m-|S|}$ coordinates associated with the protocol $\pi_S$.
        \item $\M_S$ collectively denotes:
        \begin{itemize}
            \item $M^0_S$ denotes public randomness of $\pi_S$
            \item $M^+_S = (M^1_S, M^2_S, ...)$ denotes the transcript of communication, where $M^i_S$ indicates the message sent in round $i$.
        \end{itemize}
        \item $\mu_S$ is the desired input distributions $\mu^{2^{m-|S|}}$ associated with the protocol $\pi_S$.
        \item $Z_S, A_{S0}, A_{S1}$ to be the random variable over the randomness of $(X_{S0}, Y_{S1}, \M_S)$ for which 
        {\begin{align*} Z_S & = \adv^{\pi_S}(f^\oplus(X_S, Y_S) \mid X_{S0}, Y_{S1}, \M_S) \\
        A_{S0} & = \adv^{\pi_S}(f^\oplus(X_{S0}, Y_{S0}) \mid X_{S0}, Y_{S1}, \M_S) \\
        A_{S1} & = \adv^{\pi_S}(f^\oplus(X_{S1}, Y_{S1}) \mid X_{S0}, Y_{S1}, \M_S)
        \end{align*}}
    \item $W_S$ is an event that $Z_S \geq \alpha$.
    \item $\varepsilon_S$ is the disadvantage of $\pi_S$, denoted as $1-\adv^{\pi_S}(f^{\oplus}(X_S, Y_S))$.
    \item $I_S =$ information cost of $\pi_S$. 
    \item $p_S = \pi_S(W_S)$
    \item $\chi_S$ to be defined recursively such that $\chi_\emptyset = 1$, and $\chi_{S0} = \chi_{S1} = p_S \cdot \chi_S$ for any $S$.
    \end{itemize}
    \\
    \hline
    \end{tabular}
    \caption{A set of parameters and variables related to  $\pi_S = (X_S, Y_S, \M_S).$}
    \label{table:params_S}
\end{table}

Let us now rationalize the sequence of $\chi_S$. Intuitively, $\chi_S$ is the probability of all conditioning events that are attached with $\pi_S$. To see this, we notice that $\pi_\emptyset = \pi$ is unconditioned; i.e. is conditioned by an event with probability $1 =: \chi_\emptyset$. Next, $\pi_0$ and $\pi_1$ is a result of decomposition of $\pi \mid W$. In other words, they are attached with the conditioning event $W$ which occurs with probability $\pi(W)$ which happens to be equal to $\chi_{S0}$ and $\chi_{S1}$. We can then apply this reasoning inductively.

It turns out that the analogues of Claims \ref{clm:partial_rec}, \ref{claim:mult_adv}, \ref{clm:disadv_bounds}, and \ref{clm:prob_bounds} hold at any level of the iterative process. These results are summarized in the following lemma.

\begin{lemma} For any $S$, the following statements are true. \begin{itemize}
    \item[$(i)$] $\pi_S$ has a rectangle property with respect to $\mu_S$.
    \item[$(ii)$] $\pi_S \mid W_S$ has a partial rectangle property with respect to $(\mu_{S0}, \mu_{S1})$.
    \item[$(iii)$] $Z_S = A_{S0}A_{S1}$.
    \item[$(iv)$] $\varepsilon_{S0} + \varepsilon_{S1} \leq \frac{2}{1+\sqrt{\alpha}} \cdot \left(1 - \E(Z_S \mid W_S) \right).$
    \item[$(v)$] $p_S \geq \frac{1-\varepsilon_S - \alpha}{\E(Z_S \mid W_S) - \alpha}.$
    \item[$(vi)$] $\varepsilon_{S0} = 1-\E(A_{S0} \mid W_S)$ and $\varepsilon_{S1} = 1-\E(A_{S1} \mid W_S).$
\end{itemize}
\label{lem:1_level_decomp}
\end{lemma}

For brevity, we only sketch its proof, as it highly resembles the contents of Section \ref{subsec:one_step}.
\\

\proofsketch{We first prove statements $(i)$ by induction. The base case, where $S = \emptyset$, is trivial due to the fact that $\pi$ is a standard protocol. For the induction step, suppose that the statement is true for $|S| = k$. Then, for any $S$ with $|S|=k$, $\pi_S$ has the rectangle property with respect to $\mu_S$. Following the approach in the proof of Claim \ref{clm:partial_rec}, we can show that $\pi_{S0}$ has the rectangle property with respect to $\pi_{S0}$ and $\pi_{S1}$ has the rectangle property with respect to $\pi_{S1}$. Apply this argument to all $|S| = k$ concludes the induction step.

For $(ii)$, given $(i)$ the statement is equivalent to the first half of the proof of Claim \ref{clm:partial_rec}.

Having established $(i)$ and $(ii)$, the proofs for statements $(iii)$, $(iv)$, and $(v)$ closely follow those of Claims \ref{claim:mult_adv}, \ref{clm:disadv_bounds}, and \ref{clm:prob_bounds}, respectively. For the sake of succinctness, we omit their proofs.

Finally, $(vi)$ follows from the fact that $\adv^{\pi_{S0}}(f^{\oplus}(X_{S0}, Y_{S0})) = \E(A_{S0} \mid W_S)$ by the same reasoning we argued in the earlier subsection.
}

\subsection{Related Bounds} 

In this subsection, we present several useful inequality bounds regarding the parameters we have set. These bounds shall be used repeatedly throughout the section.

\begin{claim}
For any $S$, we have $\E(Z_S \mid W_S) \geq 1-\varepsilon_S$. 
\label{clm:triangle_advantages}
\end{claim}
\begin{proof} Consider the following calculation:
\begin{align*}
    \E(Z_S) & = \sum_{(X_{S0}, Y_{S1}, \M_S)} \pi_S(X_{S0}, Y_{S1}, \M_S) \cdot \adv^{\pi_S}(f^{\oplus}(X_S, Y_S) \mid X_{S0} Y_{S1} \M_S) \\
    & = \sum_{(X_{S0}, Y_{S1}, \M_S)} \pi_S(X_{S0}, Y_{S1}, \M_S) \cdot \left| 2 \cdot \pi_S(f^{\oplus}(X_S, Y_S) = 0 \mid X_{S0} Y_{S1} \M_S) - 1\right| \\
    & \geq \left|2 \left(\sum_{(X_{S0}, Y_{S1}, \M_S)}\pi_S(X_{S0}, Y_{S1}, \M_S) \cdot \pi_S(f^{\oplus}(X_S, Y_S) = 0 \mid X_{S0} Y_{S1} \M_S) 
    \right) -1 \right| \\
    & = \left|2 \cdot \pi_S(f^{\oplus}(X_S, Y_S) = 0) -1 \right| \\
    & = \adv^{\pi_S} (f^{\oplus}(X_S, Y_S)) \\
    & = 1-\varepsilon_S
\end{align*}
where the only inequality uses the triangle inequality. Then, we can derive
\begin{align*}
    1 - \varepsilon_S & \leq \E(Z_S) \\
    & = \Pr(W_S) \cdot \E(Z_S \mid W_S) + \Pr(\overline{W_S}) \cdot \E(Z_S \mid \overline{W_S}) \\
    & \leq \Pr(W_S) \cdot \E(Z_S \mid W_S) + \Pr(\overline{W_S}) \cdot \E(Z_S \mid W_S) \\
    & = \E(Z_S \mid W_S).
\end{align*}
where the inequality follows from the notation of $W_S = (Z_S \geq \alpha)$; thus, $\E(Z_S \mid \overline{W_S}) \leq \alpha \leq \E(Z_S \mid W_S)$. This concludes the proof.
\end{proof}

\begin{corollary}
    For any $k$, we have $\sum_{|S| = k} \varepsilon_{S} \leq 2^{\tau k} \varepsilon$.
    \label{col:bound_sum_eps}
\end{corollary}
\begin{proof} Combining Lemma \ref{lem:1_level_decomp} and Claim \ref{clm:triangle_advantages}, we have $\varepsilon_{S0} + \varepsilon_{S1} \leq \frac{2}{1+\sqrt{\alpha}} \cdot \left(1-\E(Z_S \mid W_S) \right) \leq 2^{-\tau} \varepsilon_S.$ Induction on $|S|$ finishes the proof.
\end{proof}

\begin{corollary}
    For any $S$, we have $\varepsilon_S \leq \varepsilon$.
\label{cor:bound_vareps}
\end{corollary}
\begin{proof} It suffices to show that $\varepsilon_{S0}, \varepsilon_{S1} \leq \varepsilon_S$ for any $S$, as the fact follows inductively. By Lemma \ref{lem:1_level_decomp}, we have $Z_S = A_{S0}A_{S1} \leq A_{S0}$ holds pointwisely.  Moreover, by Lemma \ref{lem:1_level_decomp} and Claim \ref{clm:triangle_advantages}, the disadvantage of $\pi_{S0}$ is $\E(A_{S0} \mid W_S) \geq \E(Z_S \mid W_S) \geq 1-\varepsilon_S$. Plugging in $\varepsilon_{S0} = 1-\E(A_{S0} \mid W_S) $ yields $\varepsilon_{S0} \leq \varepsilon_S$. Similarly, we also have $\varepsilon_{S1} \leq \varepsilon_S$.
\end{proof}

The following claim is critical to several of our proofs in the next subsection.

\begin{claim} $\sum_{|S| = m} \chi_S = \Omega(n)$.
\label{clm:sum_chi}
\end{claim}

\begin{proof}
For any string $S$ with $|S| \leq m$, denote \textit{potential} of $\pi_S$ to be

$$\Lambda_S = \underbrace{\left(2^{-|S|} \log{\frac{1}{\chi_S}}\right)}_{\phi_S} \ + \ \frac{1}{\varepsilon} \cdot \underbrace{\left(\frac{1+\sqrt{\alpha}}{2}\right)^{|S|} \varepsilon_S}_{\psi_S}.$$

We claim that a conditional decomposition which splits $\pi_S$ into $\pi_{S0}$ and $\pi_{S1}$ never decreases total potentials; that is for any $S$ we must have $\Lambda_{S0} + \Lambda_{S1} \leq \Lambda_S$. To see this, consider
\begin{align*}
        \phi_{S0} + \phi_{S1} - \phi_S & = 2^{-|S|} \cdot \log \frac{\chi_S}{\sqrt{\chi_{S0} \chi_{S1}}}  \\
        &  = 2^{-|S|} \log{\frac{1}{p_S}} \tag{$\chi_{S0} = \chi_{S1} = p_S \cdot \chi_S$}\\
        & \leq 2^{-|S|} \log\left(\frac{\E(Z_S \mid W_S) - \alpha}{1-\varepsilon_S - \alpha}\right) \tag{Lemma \ref{lem:1_level_decomp}}\\
        & \leq 2^{-|S|} \cdot \frac{ \E(Z_S \mid W_S) - (1-\varepsilon_S)}{1-\varepsilon_S - \alpha} \tag{$\log x \leq x-1$}\\
        & \leq  2^{-|S|} \cdot \frac{ \E(Z_S \mid W_S) - (1-\varepsilon_S)}{\varepsilon}
    \end{align*}
where the last equality follows Claim \ref{clm:triangle_advantages} that $\E(Z_S \mid W_S) \geq 1-\varepsilon_S$, and Corollary \ref{cor:bound_vareps} that $1-\varepsilon_S - \alpha \geq 1-\varepsilon - \alpha = \varepsilon$. Moreover, \begin{align*}
        \psi_{S0} + \psi_{S1} - \psi_S & = \left(\frac{1+\sqrt{\alpha}}{2}\right)^{|S|}  \cdot \left(\frac{1+\sqrt{\alpha}}{2} \cdot (\varepsilon_{S0} + \varepsilon_{S1}) - \varepsilon_S\right) \\
        & \leq \left(\frac{1+\sqrt{\alpha}}{2}\right)^{|S|} \cdot \left(1 - \varepsilon_S - \E(Z_S \mid W_S)\right)  \tag{Lemma \ref{lem:1_level_decomp}}\\
        & \leq 2^{-|S|} \cdot \left(1 - \varepsilon_S - \E(Z_S \mid W_S)\right). \tag{Claim \ref{clm:triangle_advantages}}
    \end{align*}

Combining the two inequalities, we have $\Lambda_{S0} + \Lambda_{S1} \leq \Lambda_S$ for any $S$. Applying it recursively, we have $\Lambda_{\emptyset} \geq \sum_{|S| = m} \Lambda_S$ which leads to:
\begin{align*}
        1 = \Lambda_{\emptyset} & \geq \sum_{|S| = m} \Lambda_S = \left(\sum_{|S| = m} \frac{1}{n} \log{\frac{1}{\chi_S}}\right) + \frac{1}{\varepsilon} \cdot \left(\sum_{|S| = m} n^{-\tau} \varepsilon_S \right) \geq \sum_{|S| = m} \frac{1}{n} \log{\frac{1}{\chi_S}} .
    \end{align*}
In other words, we have $$\frac{1}{n} \sum_S {\log{\chi_S}} \geq -1.$$ Then, by the AM-GM inequality, we have

    \begin{align*}
        \sum_{S} \chi_S & \geq n \cdot \left(\prod_S \chi_S\right)^{1/n} \geq  n e^{-1}
    \end{align*}
    as wished.
\end{proof}

The following patterns will be prevalent throughout this section.

\begin{restatable}{lem}{pattern} Let $c > 0$ be an arbitrary constant. For each $S \in \{0,1\}^{\leq m}$, let $q_S \in \mathbb{R}_{\geq 0}$ satisfying the following inequality:
\begin{equation}
    q_{S0} + q_{S1} \leq \frac{1}{p_S} \cdot (q_S + c \cdot \cH(p_S)).
    \label{eq:pattern}
\end{equation}
Then, we have
$$\sum_{|S| = m} \chi_S q_S \leq q_\emptyset + O(n^{\tau} \log{n}).$$
\label{lem:pattern}
\end{restatable}

\begin{restatable}{lem}{patternconstant} Let $c, c' > 0$ be arbitrary constants. For each $S \in \{0,1\}^{\leq m}$, let $q_S \in \mathbb{R}_{\geq 0}$ satisfying the following inequality:
\begin{equation}
    q_{S0} + q_{S1} \leq \frac{1}{p_S} \cdot (q_S + c \cdot \cH(p_S)) + c'
    \label{eq:pattern_constant}
\end{equation}
Then, we have
$$\sum_{|S| = m} \chi_S q_S \leq q_\emptyset + O(n).$$
\label{lem:pattern_constant}
\end{restatable}

The proof of Lemma \ref{lem:pattern} and Lemma \ref{lem:pattern_constant} involves algebraic computations and will be deferred to the Appendix.

\subsection{Certifying a \emph{Nice} Generalized Protocol for \texorpdfstring{$f$}{f}}

We will show that following the recursive procedure (Figure \ref{figure:recursive_decomposition}), there exists an index $S \in \{0,1\}^m$ such that the protocol $\pi_S$ has small error (i.e. disdavantage), small information cost, and small $\theta$-cost with respect to $\mu$. The proof will rely on a probabilistic argument: we will show that if the index $S$ is sampled from the ``proportional'' distribution $\D$ over $\{0,1\}^m$, then $\pi_S$ has these properties in expectation. Specifically, we define the distribution $\D$ as follows.
\begin{figure}[H]
    \centering
    \begin{tabular}{|p{15cm}|}
    \hline ~\\
    \multicolumn{1}{|c|}{\textbf{A distribution $\D$ sampling an index}} \\
    \begin{itemize}
        \item Let $\D$ be a distribution over $\{0,1\}^m$ where $\D(S) = \frac{\chi_S}{\sum_{|S'| = m} \chi_{S'}}$ (i.e. is proportional to $\chi_S$.)
    \end{itemize}
    \\
    \hline
    \end{tabular}
    \caption{The ``proportional'' distribution $\D$ for sampling an index $S \in \{0,1\}^m$.}
    \label{figure:proportional_dist}
\end{figure}

We wish for these properties for the protocol $\pi_S$.

\begin{lemma}[Decomposition Lemma; informal] Over the distribution $S \sim \D$, the generalized protocol $\pi_S$  has the following properties in expectation.

\begin{enumerate}
    \item [(1)] $\pi_S$ errs with small probability $\frac{1}{\text{poly}(n)}$.
    \item[(2)] $\pi_S$ has small information cost $O(\mathcal{I}/n+1)$.
    \item[(3)] $\pi_S$ has small $\theta$-cost $\frac{1}{\text{poly}(n)}$ with respect to $\mu$.
\end{enumerate}
\label{lem:main_decomp_informal}
\end{lemma}

For the remaining of this section, we will prove a formal version of Lemma \ref{lem:main_decomp_informal}.

\paragraph{$\pi_S$ Has Small Error.} We first prove property (1). Recall that the disadvantage of $\pi_S$ is $\varepsilon_S$, which is exactly twice the error of $\pi_S$. Thus, it suffices to bound the expectation of $\varepsilon_S$.

\begin{claim} $\E(\varepsilon_S) \leq O(n^{-(1-\tau)})$.
\label{claim:eps_bounds}
\end{claim}

\begin{proof} It follows by the calculation:
\begin{align*}
    \E(\varepsilon_S) = \frac{\sum_{|S|=m} \varepsilon_S \chi_S }{\sum_{|S|=m}\chi_S} \leq \frac{\sum_{|S|=m} \varepsilon_S}{\sum_{|S|=m}\chi_S} & \leq \frac{\varepsilon n^{\tau}}{\Omega(n)} \tag{Corollary \ref{col:bound_sum_eps} and Claim \ref{clm:sum_chi}} \\
    & = O(n^{-(1-\tau)}).
\end{align*}

\end{proof}

\paragraph{$\pi_S$ Has Small Information Cost.}Next, we will prove property (2). Directly upper-bounding the information cost of $\pi_S$ can be quite complicated. Instead, we will use the $\gamma$-cost (Definition \ref{def:delta}) as an intermediate. For any subset $S \in \{0,1\}^m$, denote:
$$\Gamma_S = \gamma_{\mu_S}(\pi_S).$$
Observe that $\Gamma_{\emptyset}$ is precisely the information cost of $\pi$, which is 
$\mathcal{I}$. This is due to:
\begin{align*}\Gamma_{\emptyset} & = \gamma_{\mu^n}(\pi) \\
& = \E_{(X,Y,\M) \sim \pi} \gamma_{\mu^n, A}(\pi @ X,Y,\M) + \E_{(X,Y,\M) \sim \pi}\gamma_{\mu^n, B}(\pi @ X,Y,\M) \\
& = \E_{(X,Y,\M) \sim \pi} \log \frac{\pi(X \mid Y \M)}{\mu^n(X \mid Y)} + \E_{(X,Y,\M) \sim \pi} \log \frac{\pi(Y \mid X \M)}{\mu^n(Y \mid X)} \\
& = \E_{(X,Y,\M) \sim \pi} \log \frac{\pi(X \mid Y \M)}{\pi(X \mid Y)} + \E_{(X,Y,\M) \sim \pi} \log \frac{\pi(Y \mid X \M)}{\pi(Y \mid X)}\\
& = \mathcal{I}.
\end{align*}

The following claim argues that the $\gamma$-cost are low in expectation.

\begin{claim} $\E(\Gamma_S) = O\left(\mathcal{I}/n +1\right)$.
\end{claim}

\begin{proof} We will first show that $\Gamma_{S0} + \Gamma_{S1} \leq \frac{1}{p_S} \cdot \left[\Gamma_S + 2\cH(p_S)\right] + 4$. To do so, we compute $\gamma_{\mu_S}( \pi_S \mid W_S)$ in two different ways.

\begin{enumerate}
\item Denote the protocol $\pi_S$ by $(X_S, Y_S, \M_S)$. Recall that $\pi_S \mid W_S$ undergoes a binary decomposition into $\widetilde{\pi}_{S0} = (X_{S0}, Y_{S0}, \widetilde{\M}_{S0})$ and $\widetilde{\pi}_{S1} = (X_{S1}, Y_{S1}, \widetilde{\M}_{S1})$. Moreover, distribution-wise we have  $\pi_{S0} = (X_{S0}, Y_{S0}, \M_{S0})$ where $\M_{S0} = (\widetilde{\M}_{S0}, B_{S0})$ and $\pi_{S1} = (X_{S1}, Y_{S1}, \M_{S1})$ where $\M_{S1} = (\widetilde{\M}_{S1}, B_{S1})$.

By definition, we have:
\begin{align*} \gamma_{\mu_S}( \pi_S \mid W_S) 
= & \E_{(X_S,Y_S,\M_S ) \sim \pi_S \mid W_S } \gamma_{\mu_S, A}(\pi_S \mid W_S @ X_S,Y_S,\M_S ) \\
& \hspace{5mm} + \E_{(X_S,Y_S,\M_S ) \sim \pi_S \mid W_S } \gamma_{\mu_S, B}(\pi_S \mid W_S @ X_S,Y_S,\M_S ).
\end{align*}

Following Lemma \ref{lem:1_level_decomp}, we know that $\pi_S \mid W_S$ has the partial rectangle property with respect to $(\mu_{S0},\mu_{S1})$. The first term, via the linearlity of pointwise-$\gamma$-cost (Lemma \ref{lem:pointwise_linear_delta}), has become:
\begin{align*}
    & \E_{(X_S,Y_S,\M_S ) \sim \pi_S \mid W_S } \gamma_{\mu_S, A}(\pi_S \mid W_S @ X_S,Y_S,\M_S ) \\
    & = \E_{(X_S,Y_S,\M_S ) \sim \pi_S \mid W_S} \gamma_{\mu_{S0}, A}(\widetilde{\pi}_{S0} @ X_{S0}, Y_{S0}, \widetilde{M}_{S0}) + \E_{(X_S,Y_S,\M_S ) \sim \pi_S \mid W_S} \gamma_{\mu_{S1}, A}(\widetilde{\pi}_{S1} @ X_{S1}, Y_{S1}, \widetilde{M}_{S1}) \\
     & = \E_{(X_{S0},Y_{S0},\widetilde{\M}_{S0}) \sim \widetilde{\pi}_{S0}} \log \frac{\widetilde{\pi}_{S0}(X_{S0} \mid Y_{S0} \widetilde{\M}_{S0})}{\mu_{S0}(X_{S0} \mid Y_{S0})} +
     \E_{(X_{S1},Y_{S1},\widetilde{\M}_{S1}) \sim \widetilde{\pi}_{S1}} \log \frac{\widetilde{\pi}_{S1}(X_{S1} \mid Y_{S0} \widetilde{\M}_{S1})}{\mu_{S1}(X_{S1} \mid Y_{S1})}.
\end{align*}

Next, consider the following calculation.
\begin{align*}
     \Gamma_{S0} & = \gamma_{\mu_{S0}}(\pi_{S0}) \\ 
    & = \gamma_{\mu_{S0}}(\widetilde{\pi}_{S0} \odot B_{S0}) \\
    & = \E_{(X_{S0}, Y_{S0}, \widetilde{\M}_{S0}, B_{S0}) \sim \widetilde{\pi}_{S0} \odot B_{S0}} \log{\frac{\widetilde{\pi}_{S0}(X_{S0} \mid Y_{S0} \widetilde{\M}_{S0} B_{S0})}{\mu_{S0}(X_{S0} \mid Y_{S0})}} \\
    & \hspace{5mm} + \E_{(X_{S0}, Y_{S0}, \widetilde{\M}_{S0}, B_{S0}) \sim \widetilde{\pi}_{S0} \odot B_{S0}} \log{\frac{\widetilde{\pi}_{S0}(Y_{S0} \mid X_{S0} \widetilde{\M}_{S0} B_{S0})}{\mu_{S0}(Y_{S0} \mid X_{S0})}}.
\end{align*}
Let us expand on the first term:
\begin{align*}
&  \E_{(X_{S0}, Y_{S0}, \widetilde{\M}_{S0}, B_{S0}) \sim \widetilde{\pi}_{S0} \odot B_{S0}} \log{\frac{\widetilde{\pi}_{S0}(X_{S0} \mid Y_{S0} \widetilde{\M}_{S0} B_{S0})}{\mu_{S0}(X_{S0} \mid Y_{S0})}} \\ 
& = \E_{(X_{S0}, Y_{S0}, \widetilde{\M}_{S0}) \sim \widetilde{\pi}_{S0}}  \E_{B_{S0} \sim \widetilde{\pi}_{S0} \mid X_{S0}, Y_{S0}, \widetilde{\M}_{S0}} \log{\frac{\widetilde{\pi}_{S0}(X_{S0} \mid Y_{S0} \widetilde{\M}_{S0} B_{S0})}{\mu_{S0}(X_{S0} \mid Y_{S0})}} \\
& = \E_{(X_{S0}, Y_{S0}, \widetilde{\M}_{S0}) \sim \widetilde{\pi}_{S0}}  \E_{B_{S0} \sim \widetilde{\pi}_{S0} \mid X_{S0}, Y_{S0}, \widetilde{\M}_{S0}}  \left( \log{\frac{\widetilde{\pi}_{S0}(X_{S0} \mid Y_{S0} \widetilde{\M}_{S0})}{\mu_{S0}(X_{S0} \mid Y_{S0})}} +  \log{\frac{\widetilde{\pi}_{S0}(B_{S0} \mid X_{S0} Y_{S0} \widetilde{\M}_{S0})}{\widetilde{\pi}_{S0}(B_{S0} \mid X_{S0} Y_{S0})}}\right) \\
& = \left(\E_{(X_{S0}, Y_{S0}, \widetilde{\M}_{S0}) \sim \widetilde{\pi}_{S0}}   \log{\frac{\widetilde{\pi}_{S0}(X_{S0} \mid Y_{S0} \widetilde{\M}_{S0})}{\mu_{S0}(X_{S0} \mid Y_{S0})}}  \right) \\
& \hspace{5mm} + \left(\E_{(X_{S0}, Y_{S0}, \widetilde{\M}_{S0}) \sim \widetilde{\pi}_{S0}}  \E_{B_{S0} \sim \widetilde{\pi}_{S0} \mid X_{S0}, Y_{S0}, \widetilde{\M}_{S0}}  \log{\frac{\widetilde{\pi}_{S0}(B_{S0} \mid X_{S0} Y_{S0} \widetilde{\M}_{S0})}{\widetilde{\pi}_{S0}(B_{S0} \mid X_{S0} Y_{S0})}}\right) \\
& = \left(\E_{(X_{S0}, Y_{S0}, \widetilde{\M}_{S0}) \sim \widetilde{\pi}_{S0}}   \log{\frac{\widetilde{\pi}_{S0}(X_{S0} \mid Y_{S0} \widetilde{\M}_{S0})}{\mu_{S0}(X_{S0} \mid Y_{S0})}}  \right) + \left(\E_{(X_{S0}, Y_{S0}, \widetilde{\M}_{S0}) \sim \widetilde{\pi}_{S0}}  I(B_{S0}:\widetilde{\M}_{S0} \mid X_{S0}Y_{S0})\right) \\
& \leq \left(\E_{(X_{S0}, Y_{S0}, \widetilde{\M}_{S0}) \sim \widetilde{\pi}_{S0}}   \log{\frac{\widetilde{\pi}_{S0}(X_{S0} \mid Y_{S0} \widetilde{\M}_{S0})}{\mu_{S0}(X_{S0} \mid Y_{S0})}}  \right) + 1
\end{align*}
where the inequality follows the fact that $I(B_{S0}:M_{S0} \mid X_{S0}Y_{S0}) \leq |B_{S0}| = 1$.

Combining with the other symmetric terms, we have
$$ \Gamma_{S0} + \Gamma_{S1} \leq \gamma_{\mu_S}( \pi_S \mid W_S) + 4.$$

\item Recall that 
$$\gamma_{\mu_S}( \pi_S \mid W_S) = \E_{(X_S,Y_S,\M_S) \sim \pi_S \mid W_S} \log \frac{\pi_S(X_S \mid Y_S \M_S W_S)}{\mu_S(X_S \mid Y_S)} + \E_{(X_S,Y_S,\M_S) \sim \pi_S \mid W_S} \log \frac{\pi_S(Y_S \mid X_S \M_S W_S)}{\mu_S(Y_S \mid X_S)}.$$ For the first term, we have
\begin{align*}
    & \E_{(X_S,Y_S,\M_S) \sim \pi_S \mid W_S} \log \frac{\pi_S(X_S \mid Y_S \M_S W_S)}{\mu_S(X_S \mid Y_S)}  \\
    & = \E_{(Y_S,\M_S) \sim \pi_S \mid W_S} \dkl{\pi_S(X_S \mid Y_S \M_S W_S)}{\mu_S(X_S \mid Y_S)} \\
    & \leq  \frac{1}{p_S} \cdot \E_{(Y_S,\M_S) \sim \pi_S } \left[\dkl{\pi_S(X_S \mid Y_S \M_S  )}{\mu_S(X_S \mid Y_S)} + \cH(p_S)\right] \tag{Lemma \ref{lem:expected_kl}}\\
    & = \frac{1}{p_S} \cdot  \left[\E_{(X_S,Y_S,\M_S) \sim \pi_S }  \log \frac{\pi_S(X_S \mid Y_S \M_S)}{\mu_S(X_S \mid Y_S)} + \cH(p_S)\right].
\end{align*}

Combining with its symmetric term, we have
$$\gamma_{\mu_S}( \pi_S \mid W_S) \leq \frac{1}{p_S} \cdot \left[\Gamma_S + 2\cH(p_S)\right].$$
\end{enumerate}

Putting together the two calculations, we have $\Gamma_{S0} + \Gamma_{S1} \leq \frac{1}{p_S} \cdot \left[\Gamma_S + 2\cH(p_S)\right] + 4$ for any $S$. By Lemma \ref{lem:pattern_constant}, we have
\begin{align*}
    \sum_{|S| = m} \chi_S \Gamma_S & \leq \Gamma_{\emptyset} + O(n) = \mathcal{I} + O(n).
\end{align*}

Recall via Claim \ref{clm:sum_chi} that $\sum_{|S| = m} \chi_S  = \Omega(n)$. Therefore, we have
$$\E(\Gamma_S) = \frac{\sum_{|S| = m} \chi_S  \Gamma_S}{\sum_{|S| = m} \chi_S } = \frac{\mathcal{I} + O(n)}{\Omega(n)} = O(\mathcal{I}/n + 1)$$
as wished.
\end{proof}

As a corollary, we derive the same upper bound for the expectation of information cost.

\begin{claim} $\E(I_S) = O(\mathcal{I}/n + 1).$
\label{claim:IC_bounds}
\end{claim}
\begin{proof} It suffices to show that for any $S$, the information cost of $\pi_S$ is upper-bounded by $\Gamma_S$ (hence the reason that we use $\Gamma_S$ as a proxy.) This is due to the following calculation.
\begin{align*}
    \Gamma_S & = \E_{(X_S,Y_S,\M_S) \sim \pi_S} \log \frac{\pi_S(X_S \mid Y_S \M_S)}{\mu_S(X_S \mid Y_S)} + \E_{(X_S,Y_S,\M_S) \sim \pi_S} \log \frac{\pi_S(Y_S \mid X_S \M_S)}{\mu_S(Y_S \mid X_S)} \\
    & =  \E_{(X_S,Y_S,\M_S) \sim \pi_S} \log \frac{\pi_S(X_S \mid Y_S \M_S)}{\pi_S(X_S \mid Y_SM^0_S)} + \E_{(X_S,Y_S,\M_S) \sim \pi_S}  \log \frac{\pi_S(Y_S \mid X_S \M_S)}{\pi_S(Y_S \mid X_SM^0_S)} \\
    & \hspace{5mm} + \E_{(X_S,Y_S,\M_S) \sim \pi_S}  \log \frac{\pi_S(X_S \mid Y_SM^0_S)}{\mu_S(X_S \mid Y_S)}  + \E_{(X_S,Y_S,\M_S) \sim \pi_S}  \log \frac{\pi_S(Y_S \mid X_SM^0_S)}{\mu_S(Y_S \mid X_S)} \\
    & = I(M_S:X_S \mid Y_SM^{0}_S)  + I(M_S:Y_S \mid X_SM^{0}_S) \\
    & \hspace{5mm} + \E_{(X_S,Y_S,\M_S) \sim \pi_S}  \log \frac{\pi_S(X_S \mid Y_SM^0_S)}{\mu_S(X_S \mid Y_S)}  + \E_{(X_S,Y_S,\M_S) \sim \pi_S}  \log \frac{\pi_S(Y_S \mid X_SM^0_S)}{\mu_S(Y_S \mid X_S)} \\
    & = I_S  + \E_{(Y_S, M^{0}_S) \sim \pi_S} \dkl{\pi_S(X_S \mid Y_SM^0_S)}{\mu_S(X_S \mid Y_S)} + \E_{(X_S, M^{0}_S) \sim \pi_S} \dkl{\pi_S(Y_S \mid X_SM^0_S)}{\mu_S(Y_S \mid X_S)} \\
    & \geq I_S
\end{align*}
where the last inequality is due to the fact that KL-divergences are always non-negative.  
\end{proof}

\paragraph{$\pi_S$ Has Small $\theta$-Cost.} Finally, we prove property (3). Recall the definition of the $\theta$-cost via Definition \ref{def:theta}. For any $S$ such that $|S| \leq m$, denote  $$\Theta_S = \theta_{\mu_S}(\pi_S).$$ With this notions, we have $\Theta_\emptyset = 0$ because $\pi_\emptyset = \pi$ is a standard protocol whose input distribution is exactly $\mu^n$ (via Observation \ref{obs:std_theta_zero}.)

We will soon need the following lemma.

\begin{lemma} For any $S \in \{0,1\}^{\leq m}$ and $i \geq 1$, we have
    $$\E_{(X_S, Y_S, \M_S) \sim \pi_S \mid W_S} \log{\frac{\pi_S(M^i_S \mid X_SY_SM^{<i}_SW_S)}{\pi_S(M^i_S \mid X_SM^{<i}_SW_S)}} \leq \E_{(X_S, Y_S, \M_S) \sim \pi_S \mid W_S} \log{\frac{\pi_S(M^i_S \mid X_SY_SM^{<i}_SW_S)}{\pi_S(M^i_S \mid X_SM^{<i}_S)}}.$$
\label{lem:remove_W}
\end{lemma}
\begin{proof} Observe that the upper term of both sides are identical. Using linearlity of expectation, it is equivalent to showing that 
$$\E_{(X_S, Y_S, \M_S) \sim \pi_S \mid W_S} \log \frac{\pi_S(M^i_S \mid X_SY_SM^{<i}_SW_S)}{\pi_S(M^i_S \mid X_SM^{<i}_S)} \geq 0.$$

Consider:
\begin{align*}
    & \E_{(X_S, Y_S, \M_S) \sim \pi_S \mid W_S} \log \frac{\pi_S(M^i_S \mid X_SM^{<i}_SW_S)}{\pi_S(M^i_S \mid X_SM^{<i}_S)} \\
    & = \sum_{X_S, M^{\leq i}_S}  \pi_S(M^{\leq i}_SX_S \mid W_S) \log \frac{\pi_S(M^i_S \mid X_SM^{<i}_SW_S)}{\pi_S(M^i_S \mid X_SM^{<i}_S) }  \\
    & = \sum_{X_S, M^{\leq i}_S} \pi_S(M^{<i}_SX_S \mid W_S) \cdot \pi_S(M^i_S \mid X_SM^{<i}_SW_S) \cdot \log \frac{\pi_S(M^i_S \mid X_SM^{<i}_SW_S)}{\pi_S(M^i_S \mid X_SM^{<i}_S) } \\
    & = \sum_{X_S, M^{<i}_S}  \pi_S(M^{<i}_SX_S \mid W_S) \cdot \sum_{M^i_S} \pi_S(M^i_S \mid X_SM_{< i}W_S) \cdot \log \frac{\pi_S(M^i_S \mid X_SM^{<i}_SW_S)}{\pi_S(M^i_S \mid X_SM^{<i}_S) } \\
    & = \sum_{X_S, M^{< i}_S}  \pi_S(M^{<i}_SX_S \mid W_S) \cdot \dkl{\pi_S(M^i_S \mid X_SM^{<i}_SW_S)}{\pi_S(M^i_S \mid X_SM^{<i}_S) } \\
    & \geq 0.
\end{align*}
This concludes the proof. 
\end{proof}

The following claim argues that the $\theta$-cost are low in expectation.
\begin{claim} $\E(\Theta_S) = O(n^{-(1-\tau)} \log n)$.
\label{claim:Theta_bounds}
\end{claim}

\begin{proof} We will first show that $\Theta_{S0} + \Theta_{S1} \leq \frac{1}{p_S} \cdot \left[\Theta_S + \cH(p_S)\right]$. To do so, we calculate $\theta_{\mu_S}( \pi_S \mid W_S)$ in two different ways.
\begin{enumerate}
    \item Denote the protocol $\pi_S$ by $(X_S, Y_S, \M_S)$. Recall that $\pi_S \mid W_S$ undergoes a binary decomposition into $\widetilde{\pi}_{S0} = (X_{S0}, Y_{S0}, \widetilde{\M}_{S0})$ and $\widetilde{\pi}_{S1} = (X_{S1}, Y_{S1}, \widetilde{\M}_{S1})$. Moreover, distribution-wise we have  $\pi_{S0} = (X_{S0}, Y_{S0}, \M_{S0})$ where $\M_{S0} = (\widetilde{\M}_{S0}, B_{S0})$ and $\pi_{S1} = (X_{S1}, Y_{S1}, \M_{S1})$ where $\M_{S1} = (\widetilde{\M}_{S1}, B_{S1})$. 
    
    By Lemma \ref{lem:1_level_decomp}, we know that $\pi_S \mid W_S$ has the partial rectangle property with respect to $(\mu_{S0},\mu_{S1})$. By the linearlity of pointwise-$\theta$-cost (Lemma \ref{lem:pointwise_linear_theta}), we have
\begin{align*}
& \theta_{\mu_S}( \pi_S \mid W_S) \\
& = \E_{(X_S,Y_S,\M_S) \sim \pi_S \mid W_S } \theta_{\mu_S}(\pi_S \mid W_S @ X, Y, \M) \\
& =  \E_{(X_S,Y_S,\M_S) \sim \pi_S \mid W_S }  \theta_{\mu_{S0}}(\widetilde{\pi}_{S0}  @ X_{S0}, Y_{S0}, \widetilde{\M}_{S0}) + \E_{(X_S,Y_S,\M_S) \sim \pi_S \mid W_S} \theta_{\mu_{S1}}(\widetilde{\pi}_{S1}  @ X_{S1}, Y_{S1}, \widetilde{\M}_{S1}) \\
& =  \E_{(X_{S0},Y_{S0},\widetilde{\M}_{S0}) \sim \widetilde{\pi}_{S0} }  \theta_{\mu_{S0}}(\widetilde{\pi}_{S0}  @ X_{S0}, Y_{S0}, \widetilde{\M}_{S0}) + \E_{(X_{S1},Y_{S1},\widetilde{\M}_{S1}) \sim \widetilde{\pi}_{S1} }  \theta_{\mu_{S1}}(\widetilde{\pi}_{S1}  @ X_{S1}, Y_{S1}, \widetilde{\M}_{S1})
\end{align*}

On the other hand, we can write:
\begin{align*} \Theta_{S0} = \theta_{\mu_{S0}}(\pi_{S0}) & = \theta_{\mu_{S0}}(\widetilde{\pi}_{S0} \odot B_{S0}) \\
 & = \E_{(X_{S0}, Y_{S0}, \widetilde{\M}_{S0}, B_{S0}) \sim \widetilde{\pi}_{S0} \odot B_{S0}}  \theta_{\mu_{S0}}(\pi_{S0} @ X_{S0}, Y_{S0}, \widetilde{\M}_{S0}, B_{S0}) 
\end{align*}

By $\pi_{S0} = \widetilde{\pi}_{S0} \odot B_{S0}$, we know that the protocol $\pi_{S0}$ and $\widetilde{\pi}_{S0}$ are identical, up to the last bit $B_{S0}$ of $\pi_{S0}$. Therefore, we have 
\begin{align*}
    \theta_{\mu_{S0}}(\pi_{S0} @ X, Y, \M) - \theta_{\mu_{S0}}(\widetilde{\pi}_{S0} @ X, Y, \M) =  \log \frac{\widetilde{\pi}_{S0}(B_{S0} \mid X_{S0} Y_{S0}\M_{S0})}{\widetilde{\pi}_{S0}(B_{S0} \mid X_{S0} \M_{S0})} = 0
\end{align*}
where the last equality follows Fact \ref{obs:bs_ind} that $B_{S0} \perp Y_{S0} \mid X_{S0}\widetilde{\M}_{S0}$ in the distribution $\widetilde{\pi}_{S0}$.

Combining with the symmetric terms, we have $$\theta_{\mu_S}( \pi_S \mid W_S) = \Theta_{S0} + \Theta_{S1}.$$

\item As a result of Lemma \ref{lem:remove_W}, we have 
    \begin{align*}
        & \sum_{\text{odd } i}\E_{(X_S, Y_S, \M_S) \sim \pi_S \mid W_S} \log{\frac{\pi_S(M^{i}_S \mid X_SY_SM^{<i}_S W_S)}{\pi_S(M^{i}_S \mid X_SM^{<i}_S W_S)}} + \sum_{\text{even } i}\E_{(X_S, Y_S, \M_S) \sim \pi_S \mid W_S} \log{\frac{\pi_S(M^{i}_S \mid X_SY_SM^{<i}_S W_S)}{\pi_S(M^{i}_S \mid Y_SM^{<i}_S W_S)}} \\
        & \leq \sum_{\text{odd } i}\E_{(X_S, Y_S, \M_S) \sim \pi_S \mid W_S} \log{\frac{\pi_S(M^{i}_S \mid X_SY_SM^{<i}_SW_S)}{\pi_S(M^{i}_S \mid X_SM^{<i}_S)}} + \sum_{\text{even } i}\E_{(X_S, Y_S, \M_S) \sim \pi_S \mid W_S} \log{\frac{\pi_S(M^{i}_S \mid X_SY_SM^{<i}_S W_S)}{\pi_S(M^{i}_S \mid Y_SM^{<i}_S)}}
    \end{align*}
By patching $\E_{(X_S, Y_S, \M_S) \sim \pi_S \mid W_S} \log\frac{\pi_S(X_S, Y_S \mid W_SM^{0}_S)}{\mu_S(X_S,Y_S)}$ into both sides, we get:
$$\theta_{\mu_S}(\pi_S \mid W_S) \leq \dkl{\pi_S(X_S,Y_S,\M_S \mid W_S)}{\upsilon(X_S,Y_S,\M_S )}$$
where $\upsilon = \stdz(\pi_S, \mu_S)$ is the standardization of $\pi_S$ with respect to $\mu_S$. Using Lemma \ref{lem:expected_kl_helper}, we can further bound
\begin{align*}
    \theta_{\mu_S}(\pi_S \mid W_S) & \leq \dkl{\pi_S(X_S,Y_S,\M_S  \mid W_S)}{\upsilon(X_S,Y_S,\M_S )} \\
    & \leq \frac{1}{p_S} \cdot \left[\dkl{\pi_S(X_S,Y_S,\M_S)}{\upsilon(X_S,Y_S,\M_S)} + \cH(p_S)\right] \tag{Lemma \ref{lem:expected_kl_helper}}\\
    & = \frac{1}{p_S} \cdot \left[\theta_{\mu_S}(\pi_S) + \cH(p_S)\right] \tag{Fact \ref{fact:theta_dkl}}\\
     & = \frac{1}{p_S} \cdot \left[\Theta_S + \cH(p_S)\right].
\end{align*}
\end{enumerate}

Putting together the two calculations,  we have $\Theta_{S0} + \Theta_{S1} \leq \frac{1}{p_S} \cdot \left[\Theta_S + \cH(p_S)\right]$ for any $S$. By Lemma \ref{lem:pattern}, we have
\begin{align*}
    \sum_{|S| = m} \chi_S \Theta_S & \leq \Theta_{\emptyset} + O(n^{\tau} \log n) = O(n^{\tau} \log n).
\end{align*}

Recall via Claim \ref{clm:sum_chi} that $\sum_{|S| = m} \chi_S  = \Omega(n)$. Therefore, we have
$$\E(\Theta_S) = \frac{\sum_{|S| = m} \chi_S  \Theta_S}{\sum_{|S| = m} \chi_S } = \frac{O(n^{\tau} \log n)}{\Omega(n)} = O(n^{-(1-\tau)} \log n)$$
as wished.
\end{proof}

As a summary of this section, we reformulate our decomposition lemma. 

\begin{lemma}[Decomposition Lemma; formal] There exists $S \in \{0,1\}^m$ such that the protocol $\pi_S$ for computing $f$ has the following properties:
\begin{enumerate}
    \item[(1)] $\pi_S$ errs with small probability: $\varepsilon_S = O(n^{-(1-\tau)})$.
    \item[(2)] $\pi_S$ has small information cost: $I_S = O(\mathcal{I}/n + 1)$.
    \item[(3)] $\pi_S$ has small $\theta$-cost with respect to $\mu$: $\theta_{\mu}(\pi_S) = O(n^{-(1-\tau)} \log n)$.
\end{enumerate}
\label{lem:main_decomp_formal}
\end{lemma}

\begin{proof} Consider sampling $S \sim \D$. Applying Markov's Inequality to Claim \ref{claim:eps_bounds}, Claim \ref{claim:IC_bounds}, and Claim \ref{claim:Theta_bounds}, each the following events occurs with probability at least $0.99$: $\varepsilon_S = O(n^{-(1-\tau)})$, $I_S = O(\mathcal{I}/n + 1)$, and $\theta_{\mu}(\pi_S) = O(n^{-(1-\tau)} \log n)$. Via the union bound, the three event occurs simultaneously with positive probability. Therefore, there must exists $S \in \{0,1\}^m$ for which all three events holds. 
\end{proof}

\section{Obtaining a Standard Protocol for $f$}
\label{sec:standardize}

Let $\pi_S$ be a protocol that satisfies Lemma \ref{lem:main_decomp_formal}. We identify two remaining issues arising from the conditioning events that associate with the protocol $\pi_S$. First, $\pi_S$ is not a standard protocol. Second, the input distribution of $\pi_S$ is no longer $\mu$; therefore, the low distributional error of $\pi_S$ is evaluated against a different input distribution $\mu'$. We address both issues simultaneously. Notably, the low $\theta$-cost of $\pi_S$ implies that both $\pi_S$ is \emph{close} to being a standard protocol and $\mu'$ is \emph{close} to $\mu$. Thus, there are hopes that we can transform $\pi_S$ into a standard protocol with the correct input distribution $\mu$, while incurring only small losses in both information cost and distributional error. It turns out that this task can be achieved via \emph{standardization} as restated below.

\standardize*

The following lemma guarantees that standardizing a generalized protocol results in only small losses in information cost and distributional error, provided that the $\theta$-cost of the protocol is small.

\begin{lemma}
Suppose that a generalized protocol $\pi = (X,Y,\M)$ has information cost $\IC(\pi) = I^{\pi}(\M:X \mid YM_0) + I^{\pi}(\M:Y \mid XM_0)$, and errs with probability $\rho$ over the input distribution $\pi(x,y)$. Let $\mu$ be another input distribution, and let $\eta = \stdz(\pi, \mu)$ be the standardization of $\pi$ with respect to $\mu$. Denote $\ell = D(\pi \ \| \ \eta)$. Then, we have:
\begin{enumerate}
                \item[(1)] the information cost of $\eta$ is at most $\IC(\pi) + 2 \cdot \cH(O(\sqrt{\ell})) + O(\sqrt{\ell}) \cdot \log{(|\mathcal{X}|\cdot|\mathcal{Y}|)}$.
    \item[(2)] $\eta$ errs with probability at most $\rho + O(\sqrt{\ell})$ over the input distribution $\mu$.
\end{enumerate}
\label{lem:standardize}
\end{lemma}
\begin{proof} By Pinsker's Inequality, we have $\|\pi(X,Y,\M) - \eta(X,Y,\M)\| \leq O(\sqrt{\ell})$. By Lemma \ref{lem:coupling_lemma}, there exists a random process that generates $(X^{\pi},Y^{\pi}, \M^{\pi}) \sim \pi(X,Y,\M)$, $(X^{\eta},Y^{\eta}, \M^{\eta}) \sim \eta(X,Y,\M)$, and crucially the probability that $(X^{\pi},Y^{\pi}, \M^{\pi}) \ne (X^{\eta},Y^{\eta}, \M^{\eta})$ is at most $\|\pi(X,Y,\M) - \eta(X,Y,\M)\| \leq O(\sqrt{\ell})$. Let $F$ be such event. 

To prove (1), consider:
\begin{align*}
    & I^{\eta}(\M : X \mid YM_0) \\
    & = \Pr(\overline{F}) \cdot I^{\eta}(\M:X \mid YM_0 \overline{F}) + \Pr(F) \cdot I^{\eta}(\M:X \mid YM_0F) \\
    & =\Pr(\overline{F}) \cdot I^{\pi}(\M:X \mid Y M_0\overline{F}) + \Pr(F) \cdot I^{\eta}(\M:X \mid YM_0F) \tag{Conditioned on $\overline{F}$, we have $\eta = \pi$}\\
    & \leq I^{\pi}(\M:X \mid YM_0) + \cH(\Pr(\overline{F})) + \Pr(F) \cdot I^{\eta}(\M:X \mid Y M_0F) \tag{Lemma \ref{lem:ic_condition}} \\
    & \leq I^{\pi}(\M:X \mid YM_0)  + \cH(O(\sqrt{\ell})) + O(\sqrt{\ell}) \cdot \log |\mathcal{X}| \tag{$I^{\eta}(\M:X \mid Y \overline{E}) \leq \cH(X) \leq \log{|\mathcal{X}|}$}
\end{align*}
Combining with a symmetric term, we will have 

$$\IC(\eta) \leq \IC(\pi) + 2 \cdot \cH(O(\sqrt{\ell})) + O(\sqrt{\ell}) \cdot \log{(|\mathcal{X}|\cdot|\mathcal{Y}|)}.$$

To prove (2), observe that
\begin{align*}
    & \Pr_{(x,y) \sim \mu} \left(\eta \text{ errs on } (x,y)\right)  = \Pr(\overline{F}) \cdot \Pr_{(x,y) \sim \mu \mid  \overline{F}} \left(\eta \text{ errs on } (x,y)\right) \ + \ \Pr(F) \cdot \Pr_{(x,y) \sim \mu  \mid  F } \left(\eta \text{ errs on } (x,y)\right) \\
    & \Pr_{(x,y) \sim \pi} \left(\pi \text{ errs on } (x,y)\right)  = \Pr(\overline{F}) \cdot \Pr_{(x,y) \sim \pi \mid  \overline{F}} \left(\pi \text{ errs on } (x,y)\right) \ + \ \Pr(F) \cdot \Pr_{(x,y) \sim \pi  \mid  F } \left(\pi \text{ errs on } (x,y)\right)
\end{align*}
Recall again that that conditioned on $\overline{F}$, we have $\eta = \pi$. Therefore,
\begin{align*}
    & \Pr_{(x,y) \sim \mu} \left(\eta \text{ errs on } (x,y)\right) \\
    & = \Pr_{(x,y) \sim \pi} \left(\pi \text{ errs on } (x,y)\right) \ + \ \Pr(F) \cdot \left[\Pr_{(x,y) \sim \mu  \mid  F } \left(\eta \text{ errs on } (x,y)\right) -\Pr_{(x,y) \sim \pi  \mid  F } \left(\pi \text{ errs on } (x,y)\right)\right] \\
    & \leq \rho + O(\sqrt{\ell}).
\end{align*}
This concludes the proof.
\end{proof}

\paragraph{Completing the Proof of Lemma \ref{lem:main_lemma}.} Let $\eta = \stdz(\pi_S, \mu)$ be the standardization of $\pi_{S}$ with respect to $\mu$. Note that by the promise of Lemma \ref{lem:main_decomp_formal}, we have $\IC(\pi_S) = I_S = O(\mathcal{I}/n + 1)$, and the its distributional error is $\rho = \frac{\varepsilon_S}{2} =  O(n^{-(1-\tau)})$. Plus, we can upper-bound $\ell$ using Fact \ref{fact:theta_dkl}:
\begin{align*}
    \ell := D(\pi_{S} \ \| \ \eta)  & = \theta_{\mu}(\pi_{S}) = O(n^{-(1-\tau)}\log n).
\end{align*}

Recall $\tau \in (0,1)$ is an absolute constant. We apply Lemma \ref{lem:standardize} to the generalized protocol $\pi_{S}$. As a result, the standard protocol $\eta = \stdz(\pi, \mu)$ has the following properties:
\begin{enumerate}
    \item $\eta$ computes $f$ and errs over input distribution with probability $\rho + O\left(\sqrt{\ell}\right) = O\left(\sqrt{n^{-(1-\tau)} \log n}\right) \leq n^{-\lambda}$
    \item $\eta$ has information cost $O\left(\frac{\mathcal{I}}{n} + 1\right)  + 2 \cdot \cH(O(\sqrt{\ell})) + O(\sqrt{\ell}) \cdot \log{(|\mathcal{X}|\cdot|\mathcal{Y}|)}\leq C \cdot \left(\frac{\mathcal{I}}{n}+ \frac{\log \left(|\mathcal{X}| \cdot |\mathcal{Y}|\right)}{n^\lambda} + 1\right)$
\end{enumerate}
for some absolute constants $\lambda \in (0,1)$ and $C>0$. This completes the proof of Lemma \ref{lem:main_lemma}.

\section{Proof of XOR Lemmas}
\label{sec:pfs_xor}

Recall our main lemma which we have proved

\mainlemma*
Now we use it to derive our main results of Theorem \ref{thm:xor_ic} and Theorem \ref{thm:xor_dist_ic}. We note that throughout our proofs (including those in the Appendix), we assume that the infima in Definition \ref{def:dist_ic} and Definition \ref{def:ic} are attained by some protocol $\pi$. This assumption can be relaxed by instead considering a sequence of protocols $\{\pi_i\}_{i \geq 1}$ with corresponding costs that converge to the infimum.

\subsection{XOR Lemma for Distributional Information Cost}

\xordistic* 

\begin{proof}
    By definition of distributional information cost, let $\pi$ be a protocol for solving $f^{\oplus n}$ over $\mu^n$ that errs with probability $1/10$ and has information cost $\mathcal{I} = \IC_{\mu^n}(f^{\oplus n}, 1/10)$. Let $\eta$ be the protocol for solving $f$ obtained from Lemma \ref{lem:main_lemma}. We then derive:
\begin{align*}
    \IC_{\mu}(f, n^{-\lambda}) & = \inf_{\pi; \Pr_{(x,y) \sim \mu}(\pi(x,y) \ne f(x,y)) \leq n^{-\lambda}} \IC(\pi) \\
    & \leq \IC(\eta) \tag{$\eta$ is a protocol satisfying the infimum conditions} \\
    & \leq C \cdot \left(\frac{\mathcal{I}}{n} + \frac{\log \left(|\mathcal{X}| \cdot |\mathcal{Y}|\right)}{n^{\lambda} } + 1\right).
\end{align*}
Rearranging the inequality completes the proofs.
\end{proof}

\subsection{XOR Lemma for Information Complexity}

We will need the following lemma to boost the success probability. Its proof will be shown in the appendix.

\begin{restatable}{lem}{boosting} Let $g$ be a $\{0,1\}$-valued function, and $\varepsilon' < \varepsilon < (2e)^{-1}-0.01$.   Then, we have we have $\IC(g, \varepsilon') \leq  O\left(\frac{\log{(1/\varepsilon')}}{\log (1/\varepsilon)}\right) \cdot \IC(g, \varepsilon)$.
\label{lem:boosting}
\end{restatable}

Now we are ready to prove Theorem \ref{thm:xor_ic}.

\xoric*

\begin{proof} Let $\mu$ be the maximizer of $\IC_{\mu}(f, n^{-\lambda})$ so that $\IC_{\mathsf{D}}(f, n^{-\lambda}) = \IC_{\mu}(f, n^{-\lambda})$. Consider

\begin{align*}
    \IC(f^{\oplus n}, 1/10) \geq \IC_{\mathsf{D}}(f^{\oplus n}, 1/10) 
    & \geq \IC_{\mu^n} (f^{\oplus n}, 1/10) \\
    &\geq Cn \cdot \left(\IC_{\mu}(f, n^{-\lambda}) - \frac{\log \left(|\mathcal{X}| \cdot |\mathcal{Y}|\right)}{n^{\lambda} } - 1\right) \\
    & = Cn \cdot \left(\IC_{\mathsf{D}}(f, n^{-\lambda}) - \frac{\log \left(|\mathcal{X}| \cdot |\mathcal{Y}|\right)}{n^{\lambda} } - 1\right) \\
\end{align*}
Let us now focus on the quantity $\IC_{\mathsf{D}}(f, n^{-\lambda})$. By Theorem \ref{thm:minimax_ic}, we have 
$\IC_{\mathsf{D}}(f, n^{-\lambda}) \geq \frac{\IC(f, 2n^{-\lambda})}{2}.$ Moreover, we know that $\IC(f, n^{-1}) \leq O(\IC(f, 2n^{-\lambda}))$ via Lemma \ref{lem:boosting}. Combining everything, we have:
\begin{align*}
    \IC(f^{\oplus n}, 1/10) \geq c_1n \cdot \left(\IC(f, n^{-1}) - \frac{\log \left(|\mathcal{X}| \cdot |\mathcal{Y}|\right)}{n^{\lambda} } - 1\right)
\end{align*}
for some absolute constant $c_1 > 0$.    
\end{proof}

The following Theorem will be proved in the Appendix.

\xoricub*

Theorem \ref{thm:xor_ic} and Theorem \ref{thm:xor_ic_ub} together establish an asymptotically tight relationship (up to vanishing additive losses) between the information complexities of $f$ and $f^{\oplus n}$.

\printbibliography
\appendix

\section{Missing Proofs}
\label{appendix:missing_proofs}

We will restate and provide the missing proofs from the earlier sections.

\couplinglemma*

\begin{proof} The random process operates as follows: we interpret the randomness as a sequence of pairs $(a_1, \rho_1), (a_2, \rho_2), \ldots$, where each pair $(a_i, \rho_i)$ is drawn uniformly from $\mathcal{A} \times [0,1]$. We then set $a$ to be $a_i$ for the smallest $i$ such that $\rho_i < \mu(a_i)$, and we set $a'$ to be $a_{i'}$ for the smallest $i'$ such that $\rho_{i'} < \mu'(a_{i'})$. It is straightforward to see that $a$ follows the distribution $\mu$ and $a'$ follows the distribution $\mu'$. Thus, the next step is to bound the probability that $a \neq a'$, which is also upper-bounded by the probability that $i \neq i'$.

For any $a \in supp(\mathcal{A})$ denote an interval  $U_a := \left[0, \max\{\mu(a), \mu'(a)\}\right]$ and $D_a := \left[\min\{\mu(a), \mu'(a)\}, \max\{\mu(a), \mu'(a)\}\right]$. Also denote 
\begin{align*}
    \mu \cup \mu' := \{ (a, U_a) \ ; \ a \in supp(\mathcal{A}) \} \hspace{1cm} \text{ and } \hspace{1cm}
    \mu \bigtriangleup \mu' := \{ (a, D_a) \ ; \ a \in supp(\mathcal{A})\}
\end{align*}
and their volumes
\begin{align*}
    vol(\mu \cup \mu')  := \sum_{a \in supp(\mathcal{A})} |U_a| \hspace{1cm} \text{ and } \hspace{1cm}
    vol(\mu \bigtriangleup \mu')  := \sum_{a \in supp(\mathcal{A})} |D_a| 
\end{align*}

We say that $(a, \rho) \in \mu \cup \mu'$ if and only if $\rho \in U_a$. Similarly, we say that $(a, \rho) \in \mu \bigtriangleup \mu'$ if and only if $\rho \in D_a$. Let $j$ be the smallest index such that $(a_j, \rho_j) \in \mu \cup \mu'$. Then, $a_1 \ne a_2$ occurs if and only if $(a_j, \rho_j) \in \mu \bigtriangleup \mu'$. Thus, we have:

\begin{align*}
\Pr(a_1 \ne a_2) & = \Pr\left[(a_j, \rho_j) \in \mu \bigtriangleup \mu' \mid (a_j, \rho_j) \in \mu \cup  \mu' \right] \\
& = \frac{vol(\mu \bigtriangleup \mu')}{vol(\mu \cup \mu')} \\
& = \frac{\sum_{a \in supp(\mathcal{A})} |D_a|}{\sum_{a \in supp(\mathcal{A})} |U_a|} \\
& = \frac{\sum_{a \in supp(\mathcal{A})} |\mu(a) - \mu'(a)|}{\sum_{a \in supp(\mathcal{A})} \max(\mu(a), \mu'(a))} \\
& \leq 2 \cdot \frac{\|\mu - \mu'\|}{\sum_{a \in supp(\mathcal{A})} \mu(a)+\mu'(a)} \tag{$p+q \leq 2\cdot \max(p,q)$} \\
& = \|\mu - \mu'\|.
\end{align*}
This concludes the proof.
\end{proof}

\ptwisetheta*

\begin{proof} Throughout this proof, for convenience, we will write $\M_{(\pi_0)}$ instead of $\M^{(\pi_0)}$ and $\M_{(\pi_1)}$ instead of $\M^{(\pi_1)}$. Recall that $\pi_0 = (X_0,Y_0, \M_{(\pi_0)})$ where $\M_{(\pi_0)} = (M^0, Y_1 \circ M^1, M^2,..., M^r)$. We then can write:
\begin{align*}
    &  \theta_{\mu_{0}}( \pi_{0} \ @ X_0,Y_0,\M^{(\pi_0)})  \\
    & = \log \left(\frac{\pi_0(X_0, Y_0, \M_{(\pi_0)})}{ \pi_0(\M^0_{(\pi_0)} ) \cdot \mu_0(X_0,Y_0) \cdot \prod_{\text{odd }i \geq 1} \pi_0( \M^{i}_{(\pi_0)} \mid X_0, \M^{<i}_{(\pi_0)} ) \cdot \prod_{\text{even }i \geq 2} \pi_0( \M^{i}_{(\pi_0)} \mid Y_0, \M^{<i}_{(\pi_0)} )}\right) \\
     & = \log \left(\frac{\pi_0(X_0, Y_0, \M^{+}_{(\pi_0)} \mid \M^{0}_{(\pi_0)} )}{ \mu_0(X_0,Y_0) \cdot \pi_0( \M^{1}_{(\pi_0)} \mid X_0, \M^{0}_{(\pi_0)} ) \cdot \prod_{\text{odd }i \geq 3} \pi( M^i \mid X_0, Y_1, M^{<i} ) \cdot \prod_{\text{even }i \geq 2} \pi( M^i \mid Y M^{<i} ) }\right) \\
    & = \log \left(\frac{\pi(X_0, Y, M^+ \mid M^0 )}{ \mu_0(X_0,Y_0) \cdot \pi(Y_1M^1 \mid X_0M^0 ) \cdot \prod_{\text{odd }i \geq 3} \pi( M^i \mid X_0 Y_1 M^{<i} ) \cdot \prod_{\text{even }i \geq 2} \pi( M^i \mid Y M^{<i} ) }\right) \\
    & = \log \left(\frac{\pi(X_0, Y, M^+ \mid M^0 )}{ \mu_0(X_0,Y_0) \cdot \pi(Y_1\mid X_0M^0 ) \cdot \prod_{\text{odd }i \geq 1} \pi( M^i \mid X_0 Y_1, M^{<i} ) \cdot \prod_{\text{even }i \geq 2} \pi( M^i \mid Y M^{<i} ) }\right)
\end{align*}
Recall that $\pi_1 = (X_1,Y_1,\M_{(\pi_1)})$ where $\M_{(\pi_1)} = (M^0 \circ X_0, M^1 , M^2,..., M^r)$. We then can write:
\begin{align*}
    & \theta_{\mu_{1}}( \pi_{1}\ @ X_1,Y_1,\M^{(\pi_1)}) \\
    & = \log \left(\frac{\pi_1(X_1, Y_1, \M_{(\pi_1)} )}{ \pi_1(\M^{0}_{(\pi_1)} ) \cdot \mu_1(X_1,Y_1) \cdot \prod_{\text{odd }i \geq 1} \pi_1( \M^{i}_{(\pi_1)} \mid X_1, \M^{<i}_{(\pi_1)} ) \cdot \prod_{\text{even }i \geq 2} \pi_1( \M^{i}_{(\pi_1)} \mid Y_0, \M^{<i}_{(\pi_1)})}\right) \\
    & = \log \left(\frac{\pi_1(X_1, Y_1, \M^{+}_{(\pi_1)} \mid \M^{0}_{(\pi_1)} )}{\mu_1(X_1,Y_1) \cdot \prod_{\text{odd }i \geq 1} \pi( M_i \mid X M_{<i} ) \cdot \prod_{\text{even }i \geq 2} \pi( M_i \mid X_0 Y_1 M_{<i} ) }\right) \\
    & = \log \left(\frac{\pi(X_1, Y_1, M^+ \mid X_0M^0 )}{\mu_1(X_1,Y_1) \cdot \prod_{\text{odd }i \geq 1} \pi( M_i \mid X M_{<i} ) \cdot \prod_{\text{even }i \geq 2} \pi( M_i \mid X_0 Y_1 M_{<i} ) }\right) 
\end{align*}
Combining them, we have

\begin{align*}
    & \theta_{\mu_{0}}( \pi_{0} \ @ X_0,Y_0,\M^{(\pi_0)})  + \theta_{\mu_{1}}( \pi_{1}\ @ X_1,Y_1,\M^{(\pi_1)}) \\
 & = \log\left(\frac{\pi(X_0, Y, M^+ \mid M^0 ) \cdot \pi(X_1, Y_1, M^+ \mid X_0M^0 )}{\mu(X,Y) \cdot \pi(Y_1\mid X_0M^0 ) \cdot \pi(M^+ \mid X_0 Y_1 M^0 ) \cdot \prod_{\text{odd }i \geq 1} \pi( M^i \mid X M^{<i} ) \cdot \prod_{\text{even }i \geq 2} \pi( M^i \mid Y M^{<i} ) } \right) \\
 & = \log\left(\frac{\pi(X_0 Y M^+ \mid M^0 ) \cdot \pi(X_1 \mid X_0Y_1 \M )}{\mu(X,Y) \cdot \prod_{\text{odd }i \geq 1 \geq 1} \pi( M^i \mid X M^{<i} ) \cdot \prod_{\text{even }i \geq 2} \pi( M^i \mid Y M^{<i} ) } \right)
\end{align*}
while by definition, we have:
\begin{align*}
    & \theta_{\mu}(\pi \ @ X, Y, M) =  \log \left(\frac{\pi(X, Y, M^+ \mid M^0)}{ \mu(X,Y) \cdot \prod_{\text{odd }i \geq 1} \pi( M^i \mid X M^{<i} ) \cdot \prod_{\text{even }i \geq 2} \pi( M^i \mid Y  M^{<i} )} \right).
\end{align*}
Thus, it suffices to show that 
$$\pi(X_0 Y M^+ \mid M^0 ) \cdot \pi(X_1 \mid X_0Y_1 \M ) = \pi(X, Y, M^+ \mid M^0 )$$
which is equivalent to
\begin{align*} X_1 \perp Y_0 \mid X_0Y_1\M
\end{align*}
which is true due to the partial rectangle property of $\pi$ via Proposition \ref{prop:partial_rec_independent}.
\end{proof}

\ptwisedelta*

\begin{proof} Recall that $\pi_0 = (X_0,Y_0, \M^{(\pi_0)})$ where $\M^{(\pi_0)} = (M^0, Y_1 \circ M^1, M^2,..., M^r)$ and $\pi_1 = (X_1,Y_1,\M^{(\pi_1)})$ where $\M^{(\pi_1)} = (M^0 \circ X_0, M^1 , M^2,..., M^r)$. We then can write:
\begin{align*}
    &   \gamma_{\mu_0, A}(\pi_0 @ X_0,Y_0, \M^{(\pi_0)}) + \gamma_{\mu_1, A}(\pi_1 @ X_1,Y_1, \M^{(\pi_1)})  \\
    & = \log \frac{\pi_0(X_0 \mid Y_0 \M^{(\pi_0)})}{\mu_0(X_0 \mid Y_0)} + \log \frac{\pi_1(X_1 \mid Y_1 \M^{(\pi_0)})}{\mu_1(X_1 \mid Y_1)} \\
    & = \log \frac{\pi(X_0 \mid Y \M)}{\mu_0(X_0 \mid Y_0)} + \log \frac{\pi(X_1 \mid X_0Y_1 \M)}{\mu_1(X_1 \mid Y_1)} \\
    & = \log \frac{\pi(X_0 \mid Y \M)}{\mu_0(X_0 \mid Y_0)} + \log \frac{\pi(X_1 \mid X_0Y \M)}{\mu_1(X_1 \mid Y_1)} \tag{$X_1 \perp Y_0 \mid X_0Y_1 \M$ via Proposition \ref{prop:partial_rec_independent}} \\
    & = \log \frac{\pi(X \mid Y\M)}{\mu(X,Y)} \\
    & = \gamma_{\mu, A}(\pi @ X,Y, \M).
\end{align*}
The same proof applies for Bob's cost.
\end{proof}

\pattern*

\begin{proof} For any $k \in\{0,...,m\}$, denote $\lambda_k = \sum_{|S| = k}\chi_S q_S$. Trivially, $\lambda_0 = q_\emptyset$. From (\ref{eq:pattern}), multiplying both sides by $\chi_{S0} = \chi_{S1} = \chi_S \cdot p_S$, we have:
$$\chi_{S0} q_{S0} + \chi_{S1} q_{S1} \leq  \chi_S q_S + c \cdot \chi_S \cdot \cH(p_S) \leq \chi_S q_S + c \cdot \cH(p_S).$$ Summing it over $|S|=k$ yields:
$$\lambda_{k+1} \leq \lambda_k + c \cdot \sum_{|S| = k}\cH(p_S),$$
thus inductively, we will have $$\sum_{|S| = m}\chi_S q_S = \lambda_m \leq q_\emptyset + c \cdot \sum_{k = 0}^{m-1} \sum_{|S| = k} \cH(p_S).$$
Therefore, it remains to show that $\sum_{k = 0}^{m-1} \sum_{|S| = k} \cH(p_S) = O(n^\lambda \log n)$.

Fix any $k$. We shall upper bound $\sum_{|S| = k} \cH(p_S)$. Recall that $\alpha = 1-2\varepsilon$. By Lemma \ref{lem:1_level_decomp}, we have
\begin{align*}
    p_s \ \geq \ \frac{1-\varepsilon_S - \alpha}{\E(Z_S \mid W_S) - \alpha} \geq \frac{1-\varepsilon_S - \alpha}{1 - \alpha} = 1 - \frac{\varepsilon_S}{1-\alpha} = 1-\frac{\varepsilon_S}{2\varepsilon}.
\end{align*}
Using Corollary \ref{col:bound_sum_eps}, we have:
    
$$2^{-k} \cdot \sum_{|S| = k} {\left(1-p_S\right)} \ \leq \ 2^{-k} \sum_{|S| = k} \frac{\varepsilon_S}{2\varepsilon} \ \leq \ 2^{-(1-\tau) k-1}$$
which is also at most $\frac{1}{2}$.
Furthermore, we can derive:
\begin{align*}
    \sum_{|S| = k} \cH(p_S) & = \sum_{|S| = k} \cH(1-p_S) \\
    & \leq 2^k \cdot \cH\left(2^{-k} \cdot \sum_{|S| =k}\left(1-p_S\right)\right) \tag{$\cH$ is concave} \\
    & \leq 2^k \cdot \cH\left(2^{-(1-\tau)k-1}\right) \tag{$\cH$ is increasing in $(0,1/2]$} \\
    & \leq 2^k \cdot 2 \cdot \cH\left(2^{-(1-\tau) k-1}\right) \tag{$\cH(x) \leq 2 \cdot p \log \frac{1}{p}$ for $p \in (0,1/2]$} \\
    & = 2^{\tau k} \cdot \left((1-\tau)k + 1\right))\\
    & \leq 2^{\tau k} \cdot (k+1).
\end{align*}

Summing this over $0 \leq k \leq m-1  = \log_2{n} - 1$ yields:
$$\sum_{k = 0}^{m-1}\sum_{|S| =k} \cH(p_S) = O( n^{\tau} \log{n})$$
as wished.
\end{proof}

\patternconstant*

\begin{proof} The proof proceeds similar to that of Lemma \ref{lem:pattern}. For any $k \in\{0,...,m\}$, denote $\lambda_k = \sum_{|S| = k}\chi_S q_S$. Trivially, $\lambda_0 = q_\emptyset$. From (\ref{eq:pattern_constant}), multiplying both sides by $\chi_{S0} = \chi_{S1} = \chi_S \cdot p_S$, we have:
\begin{align*} \chi_{S0} q_{S0} + \chi_{S1} q_{S1} & \leq \chi_S q_S + c \cdot \chi_S \cdot \cH(p_S) + c' \cdot \chi_S p_S \\
& \leq \chi_S q_S + c \cdot \cH(p_S) + c'.
\end{align*} Summing it over $|S|=k$ yields:
$$\lambda_{k+1} \leq \lambda_k + c \cdot \sum_{|S| = k}\cH(p_S) + c' \cdot 2^k,$$
thus inductively, we will have $$\sum_{|S| = m}\chi_S q_S = \lambda_m \leq q_\emptyset + c \cdot \sum_{k = 0}^{m-1} \sum_{|S| = k} \cH(p_S) + c' \cdot \sum_{k = 0}^{m-1}2^k.$$
Recall from the proof of Lemma \ref{lem:pattern} that $\sum_{k = 0}^{m-1}\sum_{|S| =k} \cH(p_S) = O( n^{\tau} \log{n})$. Furthermore, we know that $\sum_{k = 0}^{m-1}2^k < 2^m = n$. Therefore, we have $$\sum_{|S| = m}\chi_S q_S \leq q_{\emptyset} + O(n)$$
as wished.
\end{proof}

\boosting*

\begin{proof} Recall via Definition \ref{def:ic} that 
$$\IC(g, \varepsilon) = \inf_{\text{$\pi$ that errs w.p. at most $\varepsilon$ on any inputs}} \max_{\mu} \IC_\mu(\pi).$$

Let $\pi$ be the minimizer of $\IC(g,\varepsilon)$ for which $\IC(g, \varepsilon) = \max_{\mu} \IC_{\mu}(\pi)$. Consider the following protocol $\pi'$ for solving $g$ over any input pair $(X,Y)$:
\begin{enumerate}
    \item[(1)] Let $T = O\left(\frac{\log{(1/\varepsilon')}}{\log (1/\varepsilon)}\right)$ be such that $(2e\varepsilon)^{T/2} < \varepsilon' / 100$. 
    \item[(2)] Alice and Bob runs $T$ independent copies of $\pi$, and output the majority answer among those copies. Denote the $T$ independent transcripts by $(M_1,...,M_T)$
\end{enumerate}
The probability of $\pi'$ being incorrect is:
\begin{align*}
    \sum_{i = T/2}^{T} \binom{T}{i} \cdot \varepsilon^{i}(1-\varepsilon)^{T-i} & \leq \sum_{i = T/2}^{T}  \left(\frac{eT}{i}\right)^{i} \cdot \varepsilon^{i} \leq \sum_{i = T/2}^{T} (2e\varepsilon)^i \leq \frac{(2e\varepsilon)^{T/2}}{1-2e\varepsilon} \leq 100 \cdot (2e\varepsilon)^{T/2} < \varepsilon'.
\end{align*}
Thus, by definition, we have $\IC(g, \varepsilon') \leq \max_{\mu'} \IC_{\mu'}(\pi').$ Let $\mu'$ be a maximizer so that 
\begin{equation}
    \IC(g, \varepsilon') \leq \IC_{\mu'}(\pi').
\label{eq:bound_ic_epsprime}
\end{equation}Consider
\begin{align*}
    \IC_{\mu'}(\pi') & = I(M_1...M_T: X \mid Y) \ + \ I(M_1...M_T: Y \mid X)
\end{align*}
and 
\begin{align*}
    I(M_1...M_T: X \mid Y)
    & = H(M_1...M_T \mid Y) - H(M_1...M_T \mid XY) \\
    & \leq \sum_{i \in [T]} {H(M_i \mid Y)} - \sum_{i \in [T]}  H(M_i \mid XY) \\
    & =\sum_{i \in [T]} I(M_i: X \mid Y)
\end{align*}
where to obtain the inequality, the first term follows subadditivity of entropy, and the second term follows the fact that the $T$ transcripts $\{M_1,...,M_T\}$ are mutually independent conditioned on $(X,Y)$. Combining with its symmetric term, we have
\begin{align*}
\IC_{\mu'}(\pi') \leq \sum_{i \in [T]} I(M_i: X \mid Y) + I(M_i: Y\mid X) = T \cdot \IC_{\mu'}(\pi) & \leq T \cdot \max_{\mu} \IC_\mu(\pi) \\
& = T \cdot \IC(g, \varepsilon).
\end{align*}
Together with (\ref{eq:bound_ic_epsprime}) and the fact that $T = O\left(\frac{\log{(1/\varepsilon')}}{\log (1/\varepsilon)}\right)$ completes the proof.
\end{proof}

\xoricub*

\begin{proof} It suffices to show that $\IC(f^{\oplus n}, 1/10) \leq n \cdot \IC(f, (10n)^{-1})$ since $\IC(f, (10n)^{-1})\leq c_2 \cdot \IC(f, n^{-1})$ for some constant $c_2 > 0$ follows from Lemma \ref{lem:boosting}. Recall via Definition \ref{def:ic} that for any function $g$ and $\varepsilon \in (0,1)$,
$$\IC(g, \varepsilon) = \inf_{\text{$\pi$ that errs w.p. at most $\varepsilon$ on any inputs}} \max_{\mu} \IC_\mu(\pi).$$
Let $\pi$ be a minimizer protocol so that $\IC(f, (10n)^{-1}) = \max_{\mu} \IC_{\mu}(\pi)$. Let $\pi'$ be the following protocol for solving $f^{\oplus n}$ over inputs $(X_1,...,X_n, Y_1,...,Y_n)$.
\begin{enumerate}
    \item [(1)] For each $i \in [n]$, the players run $\pi$ (using fresh randomness) over an input pair $(X_i, Y_i)$ to compute their belief of $f(X_i, Y_i)$. Denote this bit by $b_i$ and denote the transcript by $M_i$.
    \item [(2)] Players output $b_1 \oplus ... \oplus b_n$.
\end{enumerate}
We first argue that $\pi'$ errs with probability at most $1/10$ on any input $(X_1,...,X_n, Y_1,...,Y_n)$. By the error guarantees of $\pi$, each $b_i$ incorrectly computes $f(X_i,Y_i)$ with probability at most $(10n)^{-1}$. Via Union Bounds, all $b_i$ is correct simultaneuosly with probability at least $9/10$, resulting in their xor being correct. Therefore, the error of $\pi'$ is at most $1/10$.

Let $\mu$ be the distribution over $(X_1,...,X_n, Y_1,...,Y_n)$ that maximizes $\IC_{\mu}(\pi')$. By definition, we have
\begin{equation*}
    \IC(f^{\oplus n}, 1/10) \leq \IC_{\mu}(\pi').
\label{eq:icub_fxor}
\end{equation*}

Now we expand

\begin{align*}
    \IC_{\mu}(\pi') & = I(M_1...M_n : X_1...X_n \mid Y_1...Y_n) \ + \ I(M_1...M_n : Y_1...Y_n \mid X_1...X_n).
\end{align*}

Furthermore, consider
\begin{align*}
    & I(M_1...M_n : X_1...X_n \mid Y_1...Y_n) \\
    & = H(M_1...M_n  \mid Y_1...Y_n) - H(M_1...M_n \mid X_1...X_n Y_1...Y_n) 
\end{align*}
For the first term, we bound:
\begin{align*}
    H(M_1...M_n  \mid Y_1...Y_n) & = \sum_{i \in [n]} H(M_i \mid Y_1...Y_n M_{<i}) \tag{Chain Rule} \\
    & \leq \sum_{i \in [n]} H(M_i \mid Y_i).    
\end{align*}
For the second term, we recall that each $M_i$ only depends on $(X_i, Y_i)$; thus, $(X_1,...,X_n,Y_1,...,Y_n)$, the $n$ transcripts $(M_1,...,M_n)$ are mutually independence. Therefore, we can bound:
\begin{align*}
    H(M_1...M_n  \mid X_1...X_nY_1...Y_n) & = \sum_{i \in [n]} H(M_i \mid X_1...X_nY_1...Y_n) \\
    & = \sum_{i \in [n]} H(M_i \mid X_i Y_i).
\end{align*}
Hence, we shall have 
\begin{align*}
     I(M_1...M_n : X_1...X_n \mid Y_1...Y_n) & \leq \sum_{i \in [n]} H(M_i \mid Y_i) -  H(M_i \mid X_i Y_i) \\
     & = \sum_{i \in [n]} I(M_i:X_i \mid Y_i).
\end{align*}

Thus, we now have
$$\IC_{\mu}(\pi') \leq \sum_{i \in [n]} I(M_i:X_i \mid Y_i) + I(M_i:Y_i \mid X_i)$$

Here for each $i\in [n]$, $M_i$ is the transcript of running $\pi$ over $(X_i, Y_i)$. Therefore, we have 
\begin{align*}
    I(M_i:X_i \mid Y_i) + I(M_i:Y_i \mid X_i)  = \IC_{\mu(X_i, Y_i)}(\pi) & \leq \max_{\mu'} \IC_{\mu'}(\pi) \\
    & = \IC(f, (10n)^{-1}). 
\end{align*}

Combining everything, we shall have:
\begin{align*}
    \IC(f^{\oplus n}, 1/10) & \leq \IC_{\mu}(\pi') \\
    & \leq \sum_{i \in [n]} I(M_i:X_i \mid Y_i) + I(M_i:Y_i \mid X_i) \\
    & \leq \sum_{i \in [n]} \IC(f, (10n)^{-1}) \\
    & = n \cdot \IC(f, (10n)^{-1}) 
\end{align*}
as wished.

\end{proof}

\end{document}